\documentclass[11pt]{article}
\usepackage[font=small, labelfont=bf]{caption}
\usepackage[margin=1in]{geometry}
\usepackage{float}
\usepackage{hyperref}
\usepackage{amsfonts}
\usepackage{mathrsfs}
\usepackage{comment}
\usepackage{amsmath}
\usepackage{amssymb}
\usepackage{amsthm}
\usepackage{amscd}
\usepackage{enumerate}
\usepackage{color}
\usepackage{mathtools}
\usepackage{tikz}
\usetikzlibrary{shapes.geometric, arrows}
\mathtoolsset{showonlyrefs}

\definecolor{DarkGreen}{rgb}{0.1,0.5,0.1}
\definecolor{DarkRed}{rgb}{0.5,0.1,0.1}
\definecolor{DarkBlue}{rgb}{0.1,0.1,0.5}
\usepackage{hyperref}
\hypersetup{
    unicode=false,          
    pdftoolbar=true,        
    pdfmenubar=true,        
    pdffitwindow=false,      
    pdfnewwindow=true,      
    colorlinks=true,       
    linkcolor=DarkBlue,          
    citecolor=DarkGreen,        
    filecolor=DarkGreen,      
    urlcolor=DarkBlue,          
    pdftitle={},
    pdfauthor={},
}

\tikzstyle{startstop} = [rectangle , text width=3cm , rounded corners, minimum width=3cm, minimum height=1cm, text centered, draw=black, fill=red!30]

\tikzstyle{arrow} = [thick, ->,>=stealth]

\newtheorem{theorem}{Theorem}[section]
\newtheorem*{theorem*}{Theorem}
\newtheorem{lemma}[theorem]{Lemma}
\newtheorem{proposition}[theorem]{Proposition}
\newtheorem{corollary}[theorem]{Corollary}
\newtheorem{conjecture}[theorem]{Conjecture}
\newtheorem{definition}[theorem]{Definition}

\newtheorem{remark}[theorem]{Remark}
\newtheorem{claim}[theorem]{Claim}
\newtheorem{example}[theorem]{Example}

\newtheorem{fact}[theorem]{Fact}

\newcommand\orange[1]{{\color{orange} #1}}

\newcommand\green[1]{{\color{green} #1}}

\newcommand\magenta[1]{{\color{magenta} #1}}

\newcommand\black[1]{{\color{black} #1}}

\newcommand{\ma}{\mathcal}

\newcommand{\s}{\subseteq}

\renewcommand{\epsilon}{\varepsilon}

\newcommand{\eps}{\epsilon}
\newcommand{\ex}{{\rm{E}}}

\newcommand{\wt}{{\rm wt}}

\newcommand{\N}{{\mathbb{N}}}

\newcommand{\F}{{\mathbb{F}}}

\DeclareMathOperator{\poly}{poly}
\newcommand\one{\textbf{1}}
\newcommand\ind[1]{^{(#1)}}

\newcommand{\cC}{\mathcal{C}}

\newcommand{\abs}[1]{\left\vert#1\right\vert}

\begin{document}
\title{Improved List-Decodability and List-Recoverability of Reed--Solomon Codes via Tree Packings
\footnote{An extended abstract will appear in the Proceedings of the 62nd IEEE Annual Symposium on Foundations of Computer
Science (FOCS'21).}}

\author{Zeyu Guo\footnote{Department of Computer Science, UT Austin. Email: zguotcs@gmail.com}
\and
Ray Li\footnote{Department of Computer Science, Stanford University. Email: rayyli@cs.stanford.edu}
\and
Chong Shangguan\footnote{Research Center for Mathematics and Interdisciplinary Sciences, Shandong University, Qingdao 266237, China, and Frontiers Science Center for Nonlinear Expectations, Ministry of Education, Qingdao 266237, China. Email: theoreming@163.com}
\and
Itzhak Tamo\footnote{Department of Electrical Engineering - Systems, Tel Aviv University. Email: zactamo@gmail.com}
\and
Mary Wootters\footnote{Departments of Computer Science and Electrical Engineering, Stanford University.  Email: marykw@stanford.edu}
}
\date{}
\maketitle
\thispagestyle{empty}
\begin{abstract}
    This paper shows that there exist Reed--Solomon (RS) codes, over \black{exponentially} large finite fields \black{in the code length}, that are combinatorially list-decodable well beyond the Johnson radius, in fact almost achieving the list-decoding capacity.
    In particular, we show that for any $\epsilon\in (0,1]$ there exist RS codes with rate $\Omega(\frac{\epsilon}{\log(1/\epsilon)+1})$ that are list-decodable from radius of $1-\epsilon$. We generalize this result to list-recovery, showing that there exist $(1 - \epsilon, \ell, O(\ell/\epsilon))$-list-recoverable RS codes with rate $\Omega\left( \frac{\epsilon}{\sqrt{\ell} (\log(1/\epsilon)+1)} \right)$. Along the way we use our techniques to  give a new proof  of a  result of Blackburn on optimal linear perfect hash matrices, and strengthen it to obtain a construction of strongly perfect hash matrices.

    To derive the results in this paper we show a surprising connection of the above problems to graph theory, and in particular to the  tree packing theorem of Nash-Williams and Tutte. We also state a new conjecture that generalizes the tree-packing theorem to hypergraphs, and show that if this conjecture holds, then there would exist RS codes that are \em optimally \em (non-asymptotically) list-decodable.



\end{abstract}

\newpage

\tableofcontents
\thispagestyle{empty}

\newpage
\setcounter{page}{1}
\section{Introduction}

Reed--Solomon (RS) codes are a classical family of error correcting codes, ubiquitous in both theory and practice.  To define an  RS code, let $\F_q$ be the finite field of size $q$, and let $1 \leq k < n \leq q$.
Fix $n$ distinct evaluation points $\alpha_1, \alpha_2, \ldots, \alpha_n \in \F_q$.
The \em $[n,k]$-Reed--Solomon code \em  over $\F_q$ with evaluation points $\alpha_1, \ldots, \alpha_n$ is defined as the set
\[ \left\{ \big(f(\alpha_1), \ldots, f(\alpha_n) \big)\,:\, f \in \F_q[x], \ \deg(f) < k \right\}. \]

RS codes attain the optimal trade-off between \em rate \em and \em distance. \em
The rate of a code $\mathcal{C} \subseteq \F_q^n$ is defined as $R = \log_q|\cC|/n$.  The rate is a number between $0$ and $1$, and the closer to $1$ the better.
The (relative) distance of a code $\cC \subseteq \F_q^n$ is defined to be $\delta(\cC) = \min_{c \neq c' \in \cC} d(c,c')$, where $d(c,c') = | \{ i \in [n] \, :\, c_i \neq c_i' \} |/n$ is relative Hamming distance.  Again, the relative distance is a number between $0$ and $1$, and the closer to $1$ the better. An $[n,k]$-RS code has rate $k/n$ and distance $(n-k+1)/n$, which is the best-possible trade-off, according to the Singleton bound \cite{singletong-bound}.

Because RS codes attain this optimal trade-off (and also because they admit efficient algorithms), they have been well-studied since their introduction in the 1960's~\cite{RS-codes}.  However, perhaps surprisingly, there is still much about them that we do not know.
One notable example is their (combinatorial)\footnote{Throughout this paper, we will study \em combinatorial \em (rather than \em algorithmic\em) list-decodability.}  \em list-decodability \em and more generally their \em list-recoverability. \em  We discuss list-decodability first, and discuss list-recoverability after that.

\paragraph{List-Decodability of RS Codes.}
List-decodability can be seen as a generalization of distance.  For $\rho \in (0,1)$ and $L \geq 1$, we say that a code $\cC \subseteq \F_q^n$ is $(\rho, L)$-list-decodable if for any $y \in \F_q^n$,
\[ |\{ c \in \cC \,:\, d(c,y) \leq \rho \}| \leq L. \]
In particular, $(\rho,1)$-list-decodability is the same as having distance greater than $2\rho$.
List-decodability was introduced by Elias and Wozencraft in the 1950's~\cite{elias, wozencraft}. By now it is an important primitive in both coding theory and theoretical computer science more broadly. In general,  larger \em list sizes \em (the parameter $L$) allow for a larger \em list-decoding radius \em (the parameter $\rho$). In this work, we will be interested in the case when $\rho = 1 - \eps$ is large.

The list-decodability of RS codes is of interest for several reasons.  First, both list-decodability and RS codes are central notions in coding theory, and the authors believe that question is interesting in its own right.  Moreover, the list-decodability of RS codes has found applications in complexity theory and pseudorandomness~\cite{cai1999hardness, STV01,LP20}.

Until recently, the best bounds available on the list-decodability of RS codes were bounds that hold generically for any code.  The \em Johnson bound \em states that any code with minimum relative distance $\delta$ is $(1 - \sqrt{1 - \delta}, qn^2 \delta)$-list-decodable over an alphabet of size $q$ (\cite{johnson1962new}, see also \cite[Theorem~7.3.3]{guruswami2019essential}).
This implies that, for any $\eps\in (0,1]$, there are RS codes that are list-decodable up to radius $1 - \eps$ (with polynomial list sizes)  that have rate $\Omega(\eps^2)$.
The celebrated Guruswami--Sudan algorithm~\cite{gurus} gives an efficient algorithm to list-decode RS codes up to the Johnson bound, but it breaks down at this point.
Meanwhile, the \em list-decoding capacity theorem \em implies that  no code (and in particular, no RS code) that is list-decodable up to radius $1 - \eps$ can have rate bounded above $\eps$, unless the list size is exponential.

There have been several works over the past decade aimed at closing the
gap between the Johnson bound (rate $\Theta(\eps^2)$) and the list-decoding capacity theorem (rate $\Theta(\eps)$).
On the negative side, it is known that \em some \em RS codes (that is, some way of choosing the evaluation points $\alpha_1, \ldots, \alpha_n$), are not list-decodable substantially beyond the Johnson bound~\cite{Ben-Sasson}.
On the positive side, Rudra and Wootters~\cite{Rudra-Wootters} showed that a random choice of evaluation points will, with high probability, yield a code that is list-decodable up to radius $1 - \eps$ with rate $O\left( \frac{ \eps }{\log^5(1/\eps) \log q} \right)$.
Unfortunately, while the dependence on $\eps$ in the rate is nearly optimal (the ``correct'' dependence should be linear in $\eps$, according to the list-decoding capacity theorem),
the $\log q$ term in the denominator means that the rate necessarily goes to zero as $n$ grows, as we must have $q \geq n$ for RS codes.
Working in a different parameter regime, Shangguan and Tamo showed that over a large alphabet, there exist RS codes of rate larger than $1/9$ that can also be list-decoded beyond the Johnson bound (and in fact, optimally)~\cite{shangguan2019combinatorial0}.  However, this result only holds for small list sizes ($L=2,3$), and in particular, for such small list sizes one cannot hope to list-decode up to a radius $1 - \eps$ that approaches $1$.
Thus, there was still a substantial gap between capacity and the best known trade-offs for list-decoding RS codes.


\paragraph{List-Recoverability of RS Codes.}
The gap between capacity and the best known trade-offs for RS codes is even more pronounced for \em list recovery, \em a generalization of list decoding.
We say that a code $\cC \subseteq \F_q^n$ is $(\rho, \ell, L)$-list-recoverable if for any $S_1, S_2, \ldots, S_n \subseteq \F_q$ with $|S_i| = \ell$,
\[ |\left\{ c \in \cC \,:\, d(c, S_1 \times S_2 \times \cdots \times S_n) \leq \rho  \right\} | \leq L. \]
Here, we extend the definition of Hamming distance to sets by denoting
\[ d(c, S_1 \times\cdots \times S_n) = \frac{1}{n} |\left\{ i \in [n] \,:\, c_i \not\in S_i \right\}|.\]
The parameter $\ell$ is called the \em input list size. \em
List-decoding is the special case of list-recovery for $\ell=1$.  List-recovery first arose in the context of list-decoding (for example, the Guruswami--Sudan algorithm mentioned above is in fact a list-recovery algorithm), but has since found applications beyond that, for example in pseudorandomness \cite{GUV09} and algorithm design \cite{doron2020high}.

Both the Johnson bound and the list-decoding capacity theorem have analogs for list-recovery.  The list-recovery Johnson bound~\cite{GS01} implies that there are RS codes of rate $\Omega(\eps^2/\ell)$ that are list-recoverable up to radius $1 - \eps$ with input list size $\ell$ and polynomial output list size.
However, the list-recovery capacity theorem implies that there are codes of rate $\Omega(\eps)$ (with \em no \em dependence on $\ell$) that achieve the same guarantee, provided that the alphabet size $q$ is sufficiently large.

Thus the gap for list-recovery (between rate $\Theta(\eps^2/\ell)$ and $\Theta(\eps)$) is even larger than that for list-decoding, and in particular the dependence on $\ell$ becomes important.  To the best of our knowledge, before our work there were \em no \em results known for RS codes that established list-recovery up to arbitrarily large radius $1 - \eps$ with a better dependence on $\ell$ than $1/\ell$.

\paragraph{Motivating question.}
Given this state of affairs, our motivating question is whether or not RS codes can be list-decoded or list-recovered up to radius $1 - \eps$ with rates $\Omega(\eps)$ (in particular, with a linear dependence on $\eps$ and no dependence on the alphabet size $q$ or the input list size $\ell$).  As outlined below, we nearly resolve this question for list-decoding and make substantial progress for list-recovery.


\paragraph{Subsequent work.}
After this paper first appeared, and
inspired by the techniques in this paper and in \cite{shangguan2019combinatorial0}, Ferber, Kwan, and Sauermann \cite{ferber2020} showed that there exist $(1-\epsilon,O(1/\epsilon))$-list-decodable RS codes with rate $\Omega(\epsilon)$ over a field size polynomial in the block length, improving our result for list-decoding. Goldberg, Shangguan, and Tamo further improved the rate of \cite{ferber2020} by showing the existence of $(1-\epsilon,O(1/\epsilon))$-list-decodable RS codes with rate approaching $\frac{\epsilon}{2-\epsilon}$ \cite{goldberg2021listdecoding}.
See Section \ref{sec:related} for more details.

\textcolor{black}{
In a recent breakthrough \cite{BGM22}, Brakensiek,  Gopi, and Makam proved that for any $R>0$ and list size $L$, there exist RS codes of rate $R$ that are $(\frac{L}{L+1}(1-R), L)$-list-decodable over exponentially large fields, confirming a conjecture of Shangguan and Tamo \cite{shangguan2019combinatorial0}.
By choosing $L=\Theta(1/\epsilon)$, this result implies the existence of $(1-R-\epsilon, O(1/\epsilon))$-list-decodable RS codes of rate $R$ over exponentially large fields and subsumes our result for list-decoding.
The techniques of \cite{BGM22} are very different from ours, and uses the resolution of the \emph{GM-MDS conjecture} \cite{dau2014existence}
proved independently by Lovett \cite{lovett2018mds} and Yildiz and Hassibi \cite{yildiz2019optimum}.
Nevertheless,  our proof establishes a novel connection between list decodability and (hyper)graph-theoretic techniques, including the Nash--William--Tutte theorem and its potential generalizations, which may have further applications.
}

Lastly, we would like to mention that based on a work of Guo and Zhang \cite{guo2023randomly}, Alrabiah, Guruswami, and Li \cite{alrabiah2023randomly} eventually showed that
there exist RS codes that achieve list-decoding capacity with essentially optimal field size, namely, there are $(1-R-\epsilon, O(1/\epsilon))$-list-decodable RS codes of rate $R$ over finite fields of order $O(n)$.

\subsection{Contributions}
Our main result establishes the list-recoverability (and in particular, the list-decodability), of RS codes up to radius $1 - \eps$, representing a significant improvement over previous work.
Our techniques build on the approach of~\cite{shangguan2019combinatorial0}; the main new technical contribution is a novel connection between list-decoding RS codes and the Nash-Williams--Tutte theorem in graph theory, which may be of independent interest.
We outline our contributions below.

\paragraph{Existence of RS codes that are near-optimally list-decodable.}
Our main theorem for list-decoding is as follows.
\begin{theorem}[RS codes with near-optimal list-decoding]\label{thm:main-LD}
There is a constant $c \geq 1$ so that the following statement holds.
For any $\eps \in (0,1]$ and any sufficiently large $n$, there exist RS codes of rate $R \geq \frac{ \eps}{c (\log(1/\eps) + 1)}$ over a large enough finite field (as an \black{exponential} function of $n$ and $\eps^{-1}$), that are $(1 - \eps, c/\eps)$-list-decodable.
\end{theorem}

As discussed above, Theorem~\ref{thm:main-LD} is stronger than the result of Rudra and Wootters~\cite{Rudra-Wootters}, in that the result of \cite{Rudra-Wootters} requires that the rate tend to zero as $n$ grows, while ours holds for constant-rate codes.
On the other hand, our result requires the field size $q$ to be quite large (see Table~\ref{tab:litreview}), which \cite{Rudra-Wootters} did not require.

Our result also differs from the result of Shangguan and Tamo~\cite{shangguan2019combinatorial0} discussed above.  Because that work focuses on small list sizes, it does not apply to list-decoding radii approaching $1$.  In contrast, we are able to list-decode up to radius $1 - \eps$.  We note that \cite{shangguan2019combinatorial0} is able to show that RS codes are exactly optimal, while we are off by logarithmic factors.  Both our work and that of \cite{shangguan2019combinatorial0} require \black{exponentially} large field sizes.

\paragraph{Generalization to list-recovery.}
Theorem~\ref{thm:main-LD} follows from a more general result about list-recovery.
Our main result is the following (see Theorem~\ref{list-recover-main} for a more detailed version).

\begin{theorem}[RS codes with list-recovery beyond the Johnson bound]\label{thm:main}
 There is a constant $c\ge 1$ such that the following statement holds.
 For any $\epsilon\in(0,1]$, any positive integer $\ell$, and any sufficiently large $n$, there exist RS codes with rate $R\ge \frac{\varepsilon}{c\sqrt{\ell}(\log(1/\epsilon)+1)}$ over a large enough finite field (as an \black{exponential} function of $n$, $\epsilon^{-1}$, and $\ell$), that are $(1 - \eps, \ell, c\ell/ \eps)$-list-recoverable.
\end{theorem}

Theorem~\ref{thm:main} establishes list-recoverability for RS codes well beyond the Johnson bound,
and in particular breaks the $1/\ell$ barrier.
To the best of our knowledge, this is the first result to do so for radius arbitrarily close to $1$, although we note that work of Lund and Potukuchi achieved a similar rate for small error radius~\cite{LP20}.
We discuss related work below in Section~\ref{sec:related} and summarize quantitative results in Table~\ref{tab:litreview}.

\renewcommand{\arraystretch}{1.4}
\begin{table*}\caption{Prior work on list-decoding and list-recovery of RS codes.  Above, $C$ refers to an absolute constant.  The ``Capacity'' results refer to the list-decoding and list-recovery capacity theorems, respectively, and are impossibility results.  Above, we assume that $q \geq n$ and that $n \to \infty$ is growing relative to $1/\eps$ and $\ell$, and that $n$ is sufficiently large.}\label{tab:litreview}
\centering
\begin{tabular}{|p{3.8cm}||c|c|c|c|}
\hline
 & Radius $\rho$ & List size $L$ & Rate $R$ & Field size $q$ \\
\hline\hline
\textbf{List-Decoding:} &&&&\\
\hline
Capacity
& $1 - \eps$ & \black{$C/\eps$} & $\leq  \eps$ & - \\
\hline
Johnson bound & $1 - \eps$ & $\poly(n)$ & $C \eps^2$ & $q \geq n$ \\
\hline
\cite{Rudra-Wootters} & $1 - \eps$ & $C/\eps$ & $\frac{C\eps}{ \log^5(1/\eps) \log(q) }$ & $q \geq Cn\log^C(n/\eps)/\eps$ \\
\hline
\cite{shangguan2019combinatorial0} & $\frac{L}{L+1}(1-R)$ & $L=2,3$ & $R$ & $q = 2^{Cn}$ \\
\hline
This work (Thm.~\ref{thm:main-LD}) & $1 - \eps$ & $C/\eps$ & $\frac{C\eps}{ \log(1/\eps) }$ & $q = \left(\frac{1}{\eps}\right)^{Cn}$ \\
\hline
\hline
\textbf{List-Recovery:} &&&&\\
\hline
Capacity  & $1 - \eps$ & \black{$C/\eps$} & $\leq  \eps$ & -  \\
\hline
Johnson bound & $1 - \eps$ & $\poly(n)$ & $\frac{C\eps^2}{\ell}$ & $q \geq n$ \\
\hline
\cite{LP20} & $\rho \leq 1 - 1/\sqrt{2}$& $C\ell$ & $\frac{C}{\sqrt{\ell} \cdot \log q}$ & $q \geq C n \sqrt{\ell} \cdot \log n$ \\
\hline
This work (Thm.~\ref{thm:main}) & $1 - \eps$ & $\frac{C\ell}{\eps}$ & $\frac{C\eps}{\sqrt{\ell} \cdot \log(1/\eps)}$ & $q = \left(\frac{\ell}{\eps} \right)^{Cn}$ \\
\hline
\end{tabular}
\end{table*}

\paragraph{Applications to perfect hashing.}
Our techniques also have an application to the construction of \em strongly perfect hash matrices\em, as detailed below.
Given a matrix and a set $S$ of its columns, a row is said to {\it separate} $S$ if, restricted to this row, these columns have distinct values. For a positive integer $t$, a matrix is said to be a {\it $t$-perfect hash matrix} if any set of $t$ distinct columns of the matrix is separated by at least one row.
Perfect hash matrices were introduced by Mehlhorn \cite{data1} in 1984 for database management, and since then they have found various applications in cryptography \cite{cyr1}, circuit design \cite{circuit}, and the design of deterministic analogs of probabilistic algorithms \cite{alog}.

Let ${\rm PHF}(n,m,q,t)$ denote a $q$-ary $t$-perfect hash matrix with $n$ rows and $m$ columns. Given $m,q,t$, determining the minimal $n$ such that there exists a ${\rm PHF}(n,m,q,t)$ is one of the major open questions in this field, and has received considerable attention (see, e.g., \cite{Blackburn1998Optimal,blackburn2000perfect,shangguan2016separating}). For any integers $t\ge 2,~k\ge 2$, and sufficiently large prime power $q$, using tools from linear algebra Blackburn \cite{Blackburn1998Optimal} constructed a ${\rm PHF}(k(t-1),q^k,q,t)$, which remains the best-known construction for such parameters so far.

Constructing perfect hash matrices is related to list-recovery and list-decoding. Indeed, if the columns of our matrix are viewed as  codewords, then the matrix is a $t$-perfect hash matrix if and only if the code is $(0, t-1,t-1)$-list-recoverable \black{(we present a proof in Claim \ref{app:matrices-codes} for completeness)}. On the way to proving our main result on list-recovery, we prove a theorem (Theorem~\ref{theorem-weaker}, which we will state later), that gives very precise bounds, but only in a restricted setting.  While this setting is too restrictive to immediately yield results on list-recovery in general, it turns out to be enough to say something interesting about perfect $t$-hash matrices. In particular, we are able to recover Blackburn's result, and extend it to a generalization of perfect hashing where every set of $t$ columns needs to be separated not just by one row but by many rows.

\begin{theorem}\label{perfect-hash-matrix}
Given integers $1\le k<n$ and $t\ge 3$, for a sufficiently large prime power $q$ \black{(as an \black{exponential} function of $n$ and $t$)}, there exists an $n\times q^k$ matrix, defined on the alphabet $\mathbb{F}_q$, such that any set of $t$ columns is separated by at least $n-k(t-1)+1$ rows.
\end{theorem}

We call a matrix with the property given by Theorem \ref{perfect-hash-matrix} a {\it strongly $t$-perfect hash matrix}; this can be viewed as an ``error-resilient'' version of perfect hash matrices.  Strongly perfect hash matrices were first introduced by the third and fourth authors of this paper for $t=3$, with a slightly different definition \cite{shangguan2020degenerate}. Indeed, Lemma 25 of \cite{shangguan2020degenerate} implies the $t=3$ case of Theorem~\ref{perfect-hash-matrix}, but it breaks down at that point. We overcome this barrier, and construct strongly $t$-perfect hash matrices for all integers $t\ge 3$. The main ingredient in our proof is a surprising connection from strongly perfect hashing to graph theory (see Section \ref{phf} for the details).  Perfect hash matrices with a similar property (i.e., every set of $t$ columns needs to be separated by more than one row) have also been studied in \cite{dougherty2019hash}, but not to our knowledge in this parameter regime.

Theorem~\ref{perfect-hash-matrix} recovers Blackburn's result by taking $n=k(t-1)$, and it establishes the new result for strongly perfect hash matrices.
Based on a result of Blackburn \cite{Blackburn1998Optimal}, we also show (Proposition~\ref{proposition}) that the hash matrix in Theorem~\ref{perfect-hash-matrix} is optimal for strongly perfect hashing, among all \em linear \em hash matrices.  Both Theorem~\ref{perfect-hash-matrix} and Proposition~\ref{proposition} are proved in Section~\ref{phf}.

\paragraph{A new connection to the Nash-Williams--Tutte theorem, and a new hypergraph Nash-Williams--Tutte conjecture.}
In order to derive our results, we build on the framework of \cite{shangguan2019combinatorial0}.
That work developed a framework to view the list-decodability of Reed--Solomon codes in terms of the singularity of \em intersection matrices \em (which we define in Section~\ref{preliminaries}).  The main new technical contribution of our work is to connect the singularity of these matrices to tree-packings in particular graphs.  This connection allows us to use the Nash-Williams--Tutte theorem from graph theory to obtain our results.   The Nash-Williams--Tutte theorem gives sufficient conditions for the existence of a large \em tree packing \em (that is, a collection of pairwise edge-disjoint spanning trees) in a graph.

We think that this connection is a contribution in its own right, and it is our hope that it will lead to further improvements to our results on Reed--Solomon codes.  In particular, we hope that it will help establish the following conjecture of \cite{shangguan2019combinatorial0}:

\begin{conjecture}[Conjecture 1.5 of \cite{shangguan2019combinatorial0}]\label{conjecture-0}
 For any $\epsilon>0$ and integers $1\le k<n$ with $\epsilon n\in\mathbb{Z}$, there exist RS codes with rate $R=\frac{k}{n}$ over a large enough (as a function of $n$ and $\epsilon$) finite field, that are list-decodable from radius $1-R-\epsilon$ and list size at most $\lceil\frac{1-R-\epsilon}{\epsilon}\rceil.$
\end{conjecture}

Conjecture~\ref{conjecture-0} is stronger than our Theorem~\ref{thm:main-LD} about list-decoding.  In particular, our theorem is near-optimal, but it is interesting mostly in the low-rate/high-noise parameter regime.  In contrast, Conjecture~\ref{conjecture-0} conjectures that there exist \em exactly \em optimal RS codes, in any parameter regime.

To encourage others to use our new connection and make progress on Conjecture~\ref{conjecture-0}, we propose a method of attack in Section~\ref{sec:conjectures}.  This outline exploits our connection to the Nash-Williams--Tutte theorem, and proceeds via a conjectured generalization of the Nash-Williams--Tutte theorem to hypergraphs: we show that establishing this hypergraph conjecture (which is stated as Conjecture~\ref{conj:weak} in Section~\ref{sec:conjectures}) would in fact establish Conjecture~\ref{conjecture-0}. In Section~\ref{sec:conjectures}, we give evidence for our hypergraph Nash-Williams--Tutte conjecture, Conjecture~\ref{conj:weak}, observing that the ``easy direction'' of the conjecture follows from the Nash-Williams--Tutte theorem, and also that a quantitative relaxation of the conjecture follows from existing work \cite{CS07, CCV09}.
As further evidence of the viability of this approach, this quantitative relaxation implies a second proof of our main list-decoding result, Theorem~\ref{thm:main-LD}, and we also sketch this proof in Section~\ref{sec:conjectures}.\footnote{This second proof does not immediately establish list-recoverability, which is why we focus on our first proof.}



\subsection{Related Work}\label{sec:related}
We briefly review related work.  See Table~\ref{tab:litreview} for a quantitative comparison to prior work.

\paragraph{List-decoding of RS codes.} Ever since the Guruswami--Sudan algorithm~\cite{gurus}, which efficiently list-decodes RS codes up to the Johnson bound, it has been open to understand the extent to which RS codes are list-decodable \em beyond \em the Johnson bound, and in particular if there are RS codes that are list-decodable all the way up to the list-decoding capacity theorem, matching the performance of completely random codes.  There have been negative results that show that \em some \em RS codes are not list-decodable to capacity~\cite{Ben-Sasson}, and others that show that even if they were, in some parameter regimes we are unlikely to find an efficient list-decoding algorithm~\cite{list-dec-discrete-log}.  The work of Rudra and Wootters, mentioned above, showed that for any code with suitably good distance, a random puncturing of that code was likely to be near-optimally list-decodable; this implies that an RS code with random evaluation points is likely to be list-decodable.  Unfortunately, as discussed above, this result requires a constant alphabet size $q$ in order to yield a constant-rate code, while RS codes necessarily have $q \geq n$.

Recently, Shangguan and Tamo~\cite{shangguan2019combinatorial0} studied the list-decodability of RS codes in a different parameter regime, namely when the list size $L$ is very small, either $2$ or $3$.
They were able to get extremely precise bounds on the rate (showing that there are RS codes that are exactly optimal), but unfortunately for such small list sizes, it is impossible for any code to be list-decodable up to radius $1 - \eps$ for small $\eps$, which is our parameter regime of interest.
Unlike the approach of \cite{Rudra-Wootters}, which applies to random puncturings of any code, the work of \cite{shangguan2019combinatorial0} targeted RS codes specifically and developed an approach via studying \em intersection matrices. \em  The reason that their approach stopped at $L=3$ was the difficulty of analyzing these intersection matrices.  We build on their approach and use techniques from graph theory---in particular, the Nash-Williams--Tutte theorem---to analyze the relevant intersection matrices beyond what \cite{shangguan2019combinatorial0} were able to do.  We discuss our approach more below in Section~\ref{sec:tech}.

\paragraph{Subsequent work on list-decoding of RS codes.}
After our work first appeared, and inspired by our approach,  Ferber, Kwan, and Sauermann \cite{ferber2020} gave a beautiful proof establishing the existence of RS codes with rate $\Omega(\epsilon)$ that are list-decodable from radius $1-\epsilon$ with list size $O(1/\epsilon)$, over a polynomially (in the code's length) large finite field.\footnote{In fact, they show something more general: if one begins with any code of sufficiently large distance over a sufficiently large alphabet, and randomly punctures it to rate $\Omega(\eps)$, the resulting code is with high probability $(1 - \eps, O(1/\eps))$ list-decodable.} In further follow-up work, Goldberg, Shangguan, and Tamo \cite{goldberg2021listdecoding} further improved the rate from $\Omega(\eps)$ to a rate approaching $\frac{\eps}{2 - \eps}$.

\black{Let us give a brief overview for the main ideas of \cite{ferber2020} and \cite{goldberg2021listdecoding}. To prove that RS codes have large list-decoding radius, the authors of \cite{ferber2020} show that    a random puncturing a full-length RS code, with high probability,  resulting    in  a code with the desired property. To do so, they prove an upper bound on the number of bad puncturings, that is, puncturings that lead to a code with a poor list-decoding radius. Their upper bound in fact holds for every code with sufficiently large minimum distance, however it does not utilize the linearity of the code. The authors of \cite{goldberg2021listdecoding} strengthened the counting argument of \cite{ferber2020} (via the code's linearity) and obtained a stronger result that improves upon \cite{ferber2020}.}

Compared with our result on the list-decodability of RS codes, the results of \cite{ferber2020} and \cite{goldberg2021listdecoding} remove the logarithmic factor in $1/\eps$, and allow for smaller alphabet sizes; additionally, their proof is much shorter. However, we believe that there are still some advantages to our approach (beyond inspiring that of \cite{ferber2020} and \cite{goldberg2021listdecoding}).  First, the result of \cite{ferber2020} does not apply to list-recovery, and while \cite{goldberg2021listdecoding} does apply to list-recovery, they do not surpass the $1/\ell$ barrier in the rate.  Second, neither \cite{ferber2020} nor \cite{goldberg2021listdecoding} fully resolve Conjecture~\ref{conjecture-0} about optimal list-decodability of RS codes.  We believe that the framework and tools developed in this paper together with Conjecture~\ref{conj:weak} provide a plausible attack method to resolve Conjecture~\ref{conjecture-0}.

\paragraph{List-recovery of RS codes.} While the Guruswami--Sudan algorithm is in fact a list-recovery algorithm, much less was known about the list-recovery of RS codes beyond the Johnson bound than was known about list-decoding.  (There is a natural extension of the Johnson bound for list-recovery, see~\cite{GS01}; for RS codes, it implies that an RS code of rate about $\eps^2 / \ell$ is list-recoverable up to radius $1 - \eps$ with input list sizes $\ell$ and polynomial output list size).
\black{Similar to list-decoding,} it is known that \em some \em RS codes are not list-recoverable beyond the Johnson bound~\cite{Guruswami-rudra-limits-list-decoding}.  However, much less was known on the positive front.
In particular, neither of the works~\cite{Rudra-Wootters,shangguan2019combinatorial0} discussed above work for list-recovery.  In a recent work, Lund and Potuchuki~\cite{LP20} have proved an analogous statement to that of \cite{Rudra-Wootters}: any code of decent distance, when randomly punctured to an appropriate length, yields with high probability a good list-recoverable code.  This implies the existence of RS codes that are list-recoverable beyond the Johnson bound.
However, in \cite{LP20} there is again a dependence on $\log(q)$ in the rate bound, meaning that for RS codes, the rate must be sub-constant.  Further, the work of \cite{LP20} only applies up to radius $\rho = 1 - 1/\sqrt{2}$, and in particular does not apply to radii $\rho = 1 - \eps$, as we study in this work.  Our results also work in the constant-$\rho$ setting of \cite{LP20}, and in that regime we show that RS codes of rate $\Omega(1/\sqrt{\ell})$ are $(\rho, \ell, O(\ell))$ list-recoverable, which improves over the result of \cite{LP20} by a factor of $\log q$ in the rate.
However, we do require the field size to be much larger than that is required by \cite{LP20} (see Table~\ref{tab:litreview}).

\paragraph{Subsequent work on list-recovery of RS codes.}
The recent work of Goldberg, Shangguan, and Tamo \cite{goldberg2021listdecoding} mentioned above builds on \cite{ferber2020}, and shows that there are RS codes of rate approaching $\frac{\eps}{1 + \ell - \eps}$ that are $(1 - \eps, \ell, L_{\eps, \ell})$-list-recoverable, for a constant $L_{\eps, \ell}$ that depends only on $\eps$ and $\ell$. Compared to our work, while \cite{goldberg2021listdecoding} improves the dependence on $\eps$ in the rate by a factor of $\log(1/\eps)$, it has a worse dependence on $\ell$, and in particular does not break the $1/\ell$ barrier that is present in the Johnson bound.

\paragraph{List-decoding and list-recovery of RS-like codes.}
There are constructions---for example, of \em folded RS codes \em and \em univariate multiplicity codes \em \cite{GR08,wang,Kop15,kopparty2018improved}---of codes that are based on RS codes and that are known to achieve list-decoding (and list-recovery) capacity, with efficient algorithms.  Our goal in this work is to study Reed--Solomon codes themselves.

\paragraph{Perfect hash matrices and strongly perfect hash matrices.} Perfect hash matrices have been studied extensively since the 1980s. There are two parameter regimes that are studied. The first is when the alphabet size $q$ is constant and the number of rows tends to infinity~\cite{nilli1994perfect,fredman1984size,korner1988new,korner1986fredman,xing2019beating}.  The second is when the number of rows is viewed as a constant, while $q$ may tend to infinity~\cite{Blackburn1998Optimal,blackburn2000perfect,shangguan2016separating}.
In both cases the strength $t$ of a perfect hash matrix is a constant.  Our work studies the second case; as mentioned above, Blackburn~\cite{blackburn2000perfect} gave an optimal construction for linear hash matrices in this parameter regime, and as a special case we obtain a second proof of Blackburn's result.

The study of strongly perfect hash matrices is relatively new \cite{shangguan2020degenerate}. The thesis \cite{dougherty2019hash} collected some recent results on a closely related topic. However, the parameters considered there are quite different from those in our paper, and to the best of our knowledge, our construction is the best known in the parameter regime we consider.
Another related notion called {\it balanced hashing} was introduced in \cite{alon2007balanced,alon2009balanced}, where, with our notation, any set of $t$ columns of a matrix needs to be separated by at least $a_1$ and at most $a_2$ rows, for some integers $a_1\le a_2$. Note that in our setting, we want every set of $t$ columns to be separated by as many rows as possible, while in the setting of balanced hashing it cannot exceed the threshold $a_2$; thus, the two settings are incomparable.

\subsection{Technical Overview}\label{sec:tech}

\paragraph{Intersection matrices.}
Our approach is centered around \em intersection matrices, \em introduced in \cite{shangguan2019combinatorial0}.  Intersection matrices and their nonsingularity\footnote{To avoid possible confusion, we would like to mention that the standard notion of nonsingluar matrices applies only to {\it square scalar} matrices. However, we are in fact dealing with {\it rectangle variable} matrices perhaps with more rows than columns. Therefore, we abuse the notion nonsingluar a little by calling a variable matrix nonsingular if it contains a square submatrix with nonzero determinant as well as the same number of columns as the original matrix.} are defined formally below in Definition~\ref{Def-matrix-general}, but we give a brief informal introduction here.  A $t$-wise intersection matrix, $M$, is defined by a collection of sets $I_1, I_2, \ldots, I_t \subseteq [n]$, and has entries that are monomials in $\mathbb{F}_q[x_1, x_2, \ldots, x_n]$. It was shown in \cite{shangguan2019combinatorial0} that if there is a counter-example to the list-decodability of a Reed--Solomon code with evaluation points $(\alpha_1, \ldots, \alpha_n)$---that is, if there exist polynomials $f_1, f_2, \ldots, f_{L+1}$ that all agree with some other polynomial $g: \mathbb{F}_q \to \mathbb{F}_q$ at many points $\alpha_i$---then there is a $(L+1)$-wise intersection matrix that becomes singular when $\alpha_i$ is plugged in for $x_i$ for all $i \in [n]$.

The set-up (both the definition of an intersection matrix and the connection to list-decoding) is most easily explained by an example.  Suppose that we are interested in list-decoding for $L=3$, and suppose that we are interested in a RS code with evaluation points $\alpha_1, \alpha_2, \ldots, \alpha_n$.  Let $f_1, f_2, f_3, f_4$ and $g$ be a counter-example to list-decoding, as above, and for $1\le j\le 4$, let $I_j = \{ i\in [n] \,:\, f_j(\alpha_i) = g(\alpha_i)\}$.

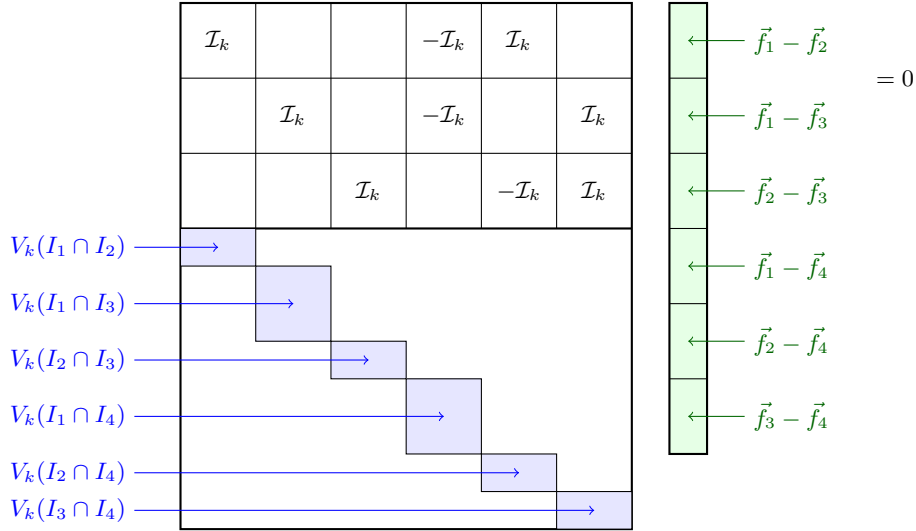
Now consider the product shown in Figure~\ref{fig:matvec}. \black{Let $f_1, f_2, f_3, f_4 \in \mathbb{F}_q[x]$ have degree $k-1$ and suppose that $I_j = \{ s \,:\, f_j(\alpha_s) = g(\alpha_s) \}$.  (In particular, $f_i$ and $f_j$ agree on $I_i \cap I_j$).  Then the matrix-vector product depicted above is zero, where the vector $\vec{f_i}$ refers to the $k$ coefficients of the polynomial $f_i$, and the $j$-th coordinate of this vector is the coefficient of $x^{j-1}$ in $f$.
Here, $V_k(I_i\cap I_j) \in \mathbb{F}_q^{|I_i\cap I_j|\times k}$ denotes the Vandermonde matrix whose rows are $[\alpha_s^0,\alpha_s^1,\dots,\alpha_s^{k-1}]$ for $s\in I_i\cap I_j$ (see \eqref{vandermonde} below for a more precise definition).
The notation $\mathcal{I}_k$ denotes the $k\times k$ identity matrix.}


\begin{figure*}
\centering
\footnotesize
\begin{tikzpicture}[scale=.5]
\draw[thick] (0,0) rectangle (12, 14);
\draw[thick] (0,8) -- (12, 8);
\foreach \j in {8, 10, 12}
{
\draw[thin] (0,\j) -- (12, \j);
}
\foreach \j in {2,4,6,8,10}
{
\draw[thin] (\j,8) -- (\j, 14);
}
\draw[fill=black!10] (0,7) rectangle (2,8);
\draw[fill=black!10] (2,5) rectangle (4, 7);
\draw[fill=black!10] (4,4) rectangle (6, 5);
\draw[fill=black!10] (6,2) rectangle (8, 4);
\draw[fill=black!10] (8,1) rectangle (10, 2);
\draw[fill=black!10] (10,0) rectangle (12, 1);
\node[black](a) at (-3, 7.5) {$V_k(I_1 \cap I_2)$};
\node[black](b) at (-3, 6) {$V_k(I_1 \cap I_3)$};
\node[black](c) at (-3, 4.5) {$V_k(I_2 \cap I_3)$};
\node[black](d) at (-3, 3) {$V_k(I_1 \cap I_4)$};
\node[black](e) at (-3, 1.5) {$V_k(I_2 \cap I_4)$};
\node[black](f) at (-3, 0.5) {$V_k(I_3 \cap I_4)$};
\draw[->,black] (a) to (1, 7.5);
\draw[->,black] (b) to (3, 6) ;
\draw[->,black] (c) to (5, 4.5);
\draw[->,black] (d) to (7, 3) ;
\draw[->,black] (e) to (9, 1.5) ;
\draw[->,black] (f) to (11, 0.5);

\foreach \j in {(1, 13), (3,11), (5,9), (9,13), (11,11), (11,9)}{
\node at \j {$\mathcal{I}_k$};
}
\foreach \j in {(7, 13), (7,11), (9,9)}{
\node at \j {$-\mathcal{I}_k$};
}

\draw[thick,fill=black!10] (13, 2) rectangle (14, 14);
\foreach \j in {12,10,8,6,4}
{
\draw (14,\j) to (13,\j);
}

\node[anchor=west,black!40!black](a) at (15, 13) {$\vec{f_1} - \vec{f_2}$};
\node[anchor=west,black!40!black](b) at (15, 11) {$\vec{f_1} - \vec{f_3}$};
\node[anchor=west,black!40!black](c) at (15, 9) {$\vec{f_2} - \vec{f_3}$};
\node[anchor=west,black!40!black](d) at (15, 7) {$\vec{f_1} - \vec{f_4}$};
\node[anchor=west,black!40!black](e) at (15, 5) {$\vec{f_2} - \vec{f_4}$};
\node[anchor=west,black!40!black](f) at (15, 3) {$\vec{f_3} - \vec{f_4}$};

\draw[black!40!black,->] (a) to (13.5, 13);
\draw[black!40!black,->] (b) to (13.5, 11);
\draw[black!40!black,->] (c) to (13.5, 9);
\draw[black!40!black,->] (d) to (13.5, 7);
\draw[black!40!black,->] (e) to (13.5, 5);
\draw[black!40!black,->] (f) to (13.5, 3);

\node at (19, 12) {$=0$};
\end{tikzpicture}

\caption{\black{An illustration of the proof for list size $L=3$}
}
\label{fig:matvec}
\end{figure*}

An inspection of Figure~\ref{fig:matvec} shows that the matrix-vector product depicted is zero.  Indeed, the top part is zero for any choice of the $f_i$, and the bottom part is zero since $f_i$ and $f_j$ are assumed to agree on $\{\alpha_s:s\in I_i \cap I_j\}$. The matrix shown is the $4$-wise intersection matrix for the sets $I_1, I_2, I_3, I_4$, evaluated at $\alpha_1, \ldots, \alpha_n$.  If the $f_i$'s agree too much with the function $g$ (i.e., if they are a counter-example to list-decodability for some given radius), then the sets $I_i \cap I_j$ are going to be larger, and this matrix will have more rows.  In particular, the more the $f_i$'s agree with $g$, the harder it is for this matrix to be singular.  Intuitively, this sets us up for a proof by contradiction: if $f_1, f_2, f_3, f_4$ agree too much with $g$, then this matrix is nonsingular (at least for a non-pathological choice of $\alpha_i$'s); but Figure~\ref{fig:matvec} displays a kernel vector!

A $t$-wise intersection matrix (for sets $I_1, \ldots, I_t$) generalizes a $4$-wise intersection matrix shown in Figure~\ref{fig:matvec}.  The bottom part looks exactly the same---a block-diagonal matrix with Vandermonde blocks---and the top part is an appropriate generalization that causes the analogous $k\cdot{t \choose 2}$-long vector corresponding to the $f_i$'s to vanish.

\paragraph{A conjecture about $t$-wise intersection matrices.}
With the motivation in Figure~\ref{fig:matvec}, the strategy of \cite{shangguan2019combinatorial0} was to study $t$-wise intersection matrices $M$ for $t = L+1$, and to show that for every appropriate choice of $I_1, \ldots, I_{t}$, the polynomial $\det(M) \in \mathbb{F}_q[x_1, x_2, \ldots, x_n]$ is not identically zero.  The list-decodability of RS codes would then follow from the DeMillo--Lipton--Schwartz--Zippel
 lemma along with a counting argument.
In particular, they made the following conjecture, and showed that it implies Conjecture~\ref{conjecture-0} about list-decoding.
Below, the \em weight \em of a family  of subsets $I_1,\ldots,I_t$ of $[n]$ is defined to be
\begin{equation}\label{definition-weight}
  {\rm wt}(I_1,\ldots,I_t)=\sum_{i=1}^t |I_i|-\left|\bigcup_{i=1}^t I_i\right|,
\end{equation}
and for a set $J$ of indices,  we use the shorthand   ${\rm wt}(I_J):={\rm wt}(I_j:j\in J)$.
\begin{conjecture}[Conjecture 5.7 of \cite{shangguan2019combinatorial0}]\label{conjecture}
  Let $t\ge 3$ be an integer and $I_1,\ldots,I_t\s[n]$ be subsets satisfying
  \begin{enumerate}
\item [\rm{(i)}] $\wt(I_J)\leq (|J|-1)k$ for all  nonempty $J\subseteq [t]$,
\item [\rm{(ii)}] Equality holds for $J=[t]$, i.e., $\wt(I_{[t]})=(t-1)k$.
  \end{enumerate}
  Then the $t$-wise intersection matrix $M_{k,(I_1,\ldots,I_t)}$ is nonsingular over any finite field.
\end{conjecture}

The conditions (i) and (ii) above turn out to be the right way of quantifying ``the $f_i$'s agree enough with $g$.''  That is, if the $f_i$'s agree too much with $g$ (in the sense of going beyond Conjecture~\ref{conjecture-0} about list-decoding), then it is possible to find sets $I_j$ so that (i) and (ii) hold.

Unfortunately, the work of \cite{shangguan2019combinatorial0} was only able to establish
Conjecture~\ref{conjecture}
for $t=3,4$ (corresponding to $L=2,3$), and it seemed challenging
to extend their techniques directly to much larger values of $L$.

\paragraph{Establishing the conjecture under an additional assumption, and using that to establish our main results.}
In this work, we use a novel connection to the Nash-Williams--Tutte theorem, which establishes the existence of pairwise edge-disjoint spanning trees in a graph, to extend the results of \cite{shangguan2019combinatorial0} to larger $L$, at the cost of an additional assumption. More precisely, we are able to show in Theorem~\ref{theorem-weaker} (stated and proved in Section~\ref{sec:3}) that Conjecture~\ref{conjecture} holds, \em provided \em that all three-wise intersections $I_i\cap I_j\cap I_\ell$ of the sets $I_j$ are empty.

The connection to the Nash-Williams--Tutte theorem is explained in Section~\ref{sec:3}.  Briefly, we consider each term in the expression
\[ \det(M) = \sum_{\sigma \in S_n} (-1)^{\mathrm{sgn}(\sigma)} \prod_{i=1}^n M_{i,\sigma(i)}.\]
We show that $\prod_{i=1}^n M_{i, \sigma(i)}$ is a nonzero monomial in $x_1, \ldots, x_n$ if and only if $\sigma$ picks out a tree packing of a graph\footnote{Throughout this paper, a tree packing of a graph $G$ means a collection of pairwise edge-disjoint spanning trees of $G$.}  that is determined by the sets $I_1, \ldots, I_t$.  It turns out that the requirements of (i) and (ii) in Conjecture~\ref{conjecture} translate exactly into the requirements needed to apply the Nash-Williams--Tutte theorem to this graph.  Thus, if (i) and (ii) hold, then there exists a tree packing in this graph and hence a nonzero term in $\det(M)$.

If the sets $I_i \cap I_j$ and $I_{i'} \cap I_{j'}$ that appear in the lower part of the $t$-wise intersection matrix do not intersect (that is, if there are no three-wise intersections among the sets $I_j$), then the reasoning above is enough to establish the conclusion of Conjecture~\ref{conjecture}, because all of the terms that appear in the expansion of the determinant are distinct monomials, and they cannot cancel.  This is why Theorem~\ref{theorem-weaker} has this assumption.

While Theorem~\ref{theorem-weaker} is not strong enough to immediately establish results for list-decoding or list-recovery (indeed, there is no reason that there should not be three-wise intersections for the polynomials $f_i$ discussed above), it \em is \em enough for our application to perfect hash matrices, which we work out in Section~\ref{phf}.

In order to apply Theorem~\ref{theorem-weaker} to list-decoding, we back off from Conjecture~\ref{conjecture} a bit. First, we allow a factor of $\Theta(\log t)$ slack on the right-hand sides of (i) and (ii). Second, rather than showing that the $t$-wise intersection matrix $M_{k,(I_1,\dots,I_t)}$ is nonsingular, we show that there exists a $t'$-wise intersection matrix that is nonsingular for some $t' < t$. Following the connection of \cite{shangguan2019combinatorial0} illustrated in Figure~\ref{fig:matvec}, this turns out to be enough to establish our main theorem on list-decoding/recovery.

We choose this smaller intersection matrix in Lemma~\ref{lem_subset} by carefully choosing a random subset $J$ of $[t]$. By greedily removing elements from the sets $\{I_j:j\in J\}$, we can obtain subsets $I_j'\subseteq I_j$ with empty three-wise intersections $I_j'\cap I_{j'}'\cap I_{j''}'=\emptyset$. Furthermore, by the careful random choice of $J$, and since we allowed a $\Theta(\log t)$ slack in the initial weight bounds, we can show this step does not delete too many elements. This is the key step of Lemma~\ref{lem_subset}.
Using some of the sets $\{I_j:j\in J\}$, we can find a smaller intersection matrix obeying the setup of Conjecture~\ref{conjecture} with the additional guarantee that all three-wise intersections are empty. We provide a more detailed summary of the proof in Section~\ref{sec:overview}.

\paragraph{Another avenue to list-decoding: a hypergraph Nash-Williams--Tutte conjecture.}
Extending our connection of list-decoding RS codes to the Nash-Williams--Tutte theorem, we show that a suitable \em hypergraph \em generalization of the Nash-Williams--Tutte theorem would imply Conjecture~\ref{conjecture} about the nonsingularity of intersection matrices, without any need for an additional assumption about three-wise intersections of the sets $I_j$.


We conjecture that such a generalization is true, and we state it in Section~\ref{sec:conjectures} as Conjecture~\ref{conj:weak}. It requires a bit of notation to set up, so we do that in Section~\ref{sec:conjectures} rather than here; however, the reader interested in the hypergraph conjecture can at this point jump straight to Section~\ref{sec:conjectures} without missing anything.

We show that if our hypergraph conjecture were true, it would imply Conjecture~\ref{conjecture}, on the nonsingularity of intersection matrices (Theorem~\ref{thm:complete-2}).  This in turn would imply
Conjecture~\ref{conjecture-0}, establishing the existence of RS codes with optimal list-decodability.  This suggests a plan of attack towards Conjecture~\ref{conjecture-0}.

While we are unable to establish this challenging conjecture in full, we give some evidence for it.
First, we show that the ``easy part'' of the conjecture follows from the Nash-Williams--Tutte theorem.
Second, we observe that a quantitative relaxation of the conjecture follows from known results on Steiner tree packings \cite{CS07} and disjoint bases of polymatroids \cite{CCV09}.
This relaxation can be combined with the connection of hypergraph packings and intersection matrices established in Theorem~\ref{thm:complete-2}, and the connection between intersection matrices and list decoding RS codes, to give a second proof of Theorem~\ref{thm:main-LD}, that there are \emph{near}-optimally list-decodable RS codes.

In addition to implying the optimal list-decodability of RS codes, Conjecture~\ref{conj:weak} may be of independent interest.
A hypergraph generalization of Nash-Williams--Tutte is known for \emph{partition-connected} hypergraphs \cite{frank2003decomposing} (see Section~\ref{sec:conjectures} for definition), a well studied notion.
However, for a different notion called \emph{weak-partition-connectivity}, less seems to be known, and Conjecture~\ref{conj:weak} poses a Nash-Williams--Tutte generalization for weakly-partition-connected hypergraphs.


\paragraph{Organization.}
A graphical overview of our results can be found in Figure~\ref{fig:M1}.
We begin in Section~\ref{preliminaries} with the needed notation and definitions, including the definition of $t$-wise intersection matrices.

In Sections~\ref{sec:3}, \ref{phf}, and \ref{list-recovery}, we prove Theorem~\ref{thm:main-LD} and Theorem~\ref{thm:main} using our proof of Conjecture~\ref{conjecture} under the additional assumption of no three-wise intersections.  More precisely,
in Section~\ref{sec:3}, we show how to use the Nash-Williams--Tutte theorem from graph theory to prove Theorem~\ref{theorem-weaker}, which establishes Conjecture~\ref{conjecture} under the assumption of no three-wise intersections.
In Section~\ref{phf} we use Theorem~\ref{theorem-weaker} to prove Theorem~\ref{perfect-hash-matrix} about perfect hash matrices. A reader interested in the list-recovery result can skip Section \ref{sec:3} and continue directly to Section \ref{list-recovery}.
In Section~\ref{list-recovery} we prove Theorem~\ref{thm:main} (and thus Theorem~\ref{thm:main-LD}) on the list-recoverability (list-decodability) of RS codes.

In Sections~\ref{sec:conjectures}, we conjecture a hypergraph generalization of the Nash-Williams--Tutte theorem and prove (Theorem~\ref{thm:complete-2}) that it implies optimal list-decodability of RS codes.
We also give some evidence for it by observing an ``easy direction'' follows from the ordinary Nash-Williams--Tutte theorem, by highlighting a known relaxation, and sketching how this relaxation gives us a second proof of Theorem~\ref{thm:main-LD}.


\begin{figure*}
\label{diagram}
\centering
\begin{tikzpicture}[node distance=2cm]
\footnotesize

\node (1)[startstop, yshift=0cm, fill=white!20]{\underline{Conjecture \ref{conj:weak}} \\ Hypergraph Nash-Williams--Tutte Conjecture};
\node (2)[startstop, right of =1, xshift=3cm,fill=white!20]{\underline{Conjecture \ref{conjecture}}\\ Nonsingularity of intersection matrices};
\node (3)[startstop, right of =2, xshift=3cm,fill=white!20]{\underline{Conjecture \ref{conjecture-0}}\\ Optimal List Decoding of RS Codes};
\draw [arrow] (1) --node[anchor=south] {Thm. \ref{thm:complete-2}}(2);
\draw [arrow] (2) --node[anchor=south] {\cite{shangguan2019combinatorial0} }(3);

\node (0)[startstop, above of=1, yshift=3cm,fill=white!20]{\underline{Lemma \ref{Nash-Williams}}\\ Nash-Williams--Tutte theorem \cite{nash1961edge,tutte1961problem}.};

\node (7)[startstop, right of=0, xshift=3cm,fill=white!20]{\underline{Theorem \ref{theorem-weaker}}\\ Nonsingularity of intersection matrices, provided that $I_i \cap I_j \cap I_\ell = \emptyset$ for all distinct $i,j,l$.};

\node (9)[startstop,  thick, right of =7, xshift=3cm,fill=white!20]{\underline{Theorem \ref{thm:main}} \\Main Theorem: List Recovery of RS codes};
\draw [arrow] (0)--(7);
\draw [arrow] (7)--(9);
\node (8)[startstop, thick, above of =9, yshift=0.5cm,fill=white!20]{\underline{Theorem \ref{perfect-hash-matrix}} \\Existence of strongly perfect hash matrices};
\draw [arrow] (7) |-(8);

\node (6)[startstop, thick, below of =9, xshift=0cm,yshift= -0.5cm,fill=white!50!white!20]{\underline{Theorem \ref{thm:main-LD}}\\ List Decoding of RS codes };
\draw [arrow] (9) --(6);

\node [left of =1,color=red!70!black,anchor=east]
{\begin{minipage}{3cm} \begin{center} Proposed roadmap to optimal list-decoding, presented in Section~\ref{sec:conjectures}.
\end{center}\end{minipage}};
\node [left of =0,color=green!50!black,anchor=east]
{\begin{minipage}{3cm}
\begin{center}
Our proof of Theorem~\ref{thm:main-LD}, presented in Sections~\ref{sec:3},~\ref{phf},~\ref{list-recovery}
\end{center}\end{minipage}};

\draw[dashed] (1) to (0);
\draw[dashed] (2) to (7);
\draw[dashed] (3) to (6);
\end{tikzpicture}

\caption{
A diagram of the results and conjectures presented in this work.
Solid arrows represent logical implications.
Dashed lines indicate how the proposed roadmap to optimal list decoding parallels our proof of Theorem~\ref{thm:main-LD}.
}\label{fig:M1}
\end{figure*}

\section{Preliminaries}\label{preliminaries}

The main goal of this section is to present the definition of $t$-wise intersection matrices over an arbitrary field $\mathbb{F}$.

Let $\N^+=\{1,2,\dots\}$ and $[n]=\{1,2,\dots, n\}$ for $n\in\N^+$. Denote by $\log x$  the base-$2$ logarithm of $x$. For a finite set $X$ and an integer $1\le k\le |X|$, let $\binom{X}{k}=\{A\subseteq X:|A|=k\}$ be the family of all $k$-subsets of $X$. For an integer $t\ge 3$, we define the following lexicographic order on $\binom{[t]}{2}$. For distinct $S_1,S_2 \in\binom{[t]}{2}$, $S_1< S_2$ if and only if $\max(S_1)<\max(S_2)$ or $\max(S_1)=\max(S_2)$ and $\min(S_1)<\min(S_2)$.
For a partition $\mathcal{P}$ of $X$, let $|\mathcal{P}|$ denote the number of parts of $\mathcal{P}$.
In the remaining part of this paper, assume that $n,k$ are integers satisfying $1\le k<n$.

We view a polynomial $f\in\mathbb{F}_q[x]$ of degree at most $k-1$ as a vector of length $k$ defined by its $k$ coefficients, where for $1\le i\le k$, the $i$-th coordinate of this vector is the coefficient of $x^{i-1}$ in $f$. By abuse of notation that vector is also denoted by $f$.

\subsection{Cycle Spaces}
We need the notion of the {\it cycle space} of a graph, which is typically defined over the Boolean field $\mathbb{F}_2$ (see, e.g.,  \cite{diestel2017graphentheory}). Here we define it over an arbitrary field $\mathbb{F}$. An equivalent definition can be found in \cite{BBNN93}, where it is called the ``circuit-subspace''.

Let $K_t$ be the undirected complete graph with the vertex set $[t]$. Denote by  $\{i,j\}$  the edge connecting vertices $i$ and $j$. Let 
\black{{\it oriented $K_t$}}
be the oriented graph obtained by replacing $\{i,j\}$ with the directed edge $(i,j)$ for all $1\le i < j\le t$.
For a \black{undirected} graph $G$ with vertex set $[t]$, an {\it oriented cycle} in $G$ is a set of directed edges of the form
\[
C=\{(i_0, i_1), (i_1, i_2), \dots, (i_{m-1}, i_m)\}
\]
 where $m\geq 3$, $i_0,\dots,i_{m-1}$ are distinct, $i_m=i_0$ and $\{i_{j-1}, i_j\}$ is an edge of $G$ for all $j=1,\dots,m$.

 Suppose $C$ is a union of edge-disjoint oriented cycles in $G$.
Then $C$ is uniquely represented by a vector $u^C=(u^C_{\{i,j\}}: \{i,j\}\in \binom{[t]}{2})\in\mathbb{F}^{\binom{t}{2}}$, defined  for $ 1\leq i<j\leq t$ by
 \[
u^C_{\{i,j\}}=\begin{cases}1 & (i,j)\in C, \\
 -1 & (j,i)\in C, \\
 0 & \text{else}.
 \end{cases}  \
 \]
 Hence, the sign of a nonzero coordinate  $u^C_{\{i,j\}}$ indicates whether the orientation of $\{i,j\}$ in $C$ complies with its orientation in \black{the oriented $K_t$}. 
We further assume that the coordinates of $u^C $ are ordered by the aforementioned lexicographic order on $\binom{[t]}{2}$.

Denote by $C(G)\subseteq \mathbb{F}^{\binom{t}{2}}$ the subspace spanned by the set of vectors  $$\black{C(G)=}\{u^C: C \text{ is an oriented cycle in } G\}$$ over $\mathbb{F}$.
We call $C(G)$ the {\it cycle space} of $G$ over $\mathbb{F}$. We are particularly interested in the cycle space $C(K_t)$ of $K_t$.
For distinct $i,j,\ell\in [t]$, denote by $\Delta_{ij\ell}$ the oriented cycle $\{(i,j), (j,\ell),(\ell, i)\}$ and call it  an {\it oriented triangle}.
We have the following lemma, generalizing \cite[Theorem 1.9.5]{diestel2017graphentheory}.

\begin{lemma} \label{cycle-space-basis}
The \black{cycle} space $C(K_t)\subseteq \mathbb{F}^{\binom{t}{2}}$ has dimension $\binom{t-1}{2}$, and the set
\[
\mathcal{B}_t=\{u^{\Delta_{ijt}}:1\le i<j\le t-1\}
\]
 is a basis of \black{the cycle space} $C(K_t)$.
\end{lemma}

\begin{proof}
\black{Note that $\mathcal{B}_t$ is formed by the set of vectors representing the oriented triangles $\Delta_{ijt}$ with the third vertex fixed to $t$.} The vectors in $\mathcal{B}_t$  are linearly independent since  $u^{\Delta_{ijt}}_{\{i,j\}}=1$ and $u^{\Delta_{ijt}}_{\{i',j'\}}=0$ for $1\leq i<j\leq t-1$ and  $1\leq i'<j'\leq t-1$ with $\{i,j\}\neq \{i',j'\}$.
Let $W$ be the span of $\mathcal{B}_t$ over $\mathbb{F}$.
Consider an arbitrary oriented cycle $C$ in $K_t$. We claim that $u^C\in W$, and this would imply that $\mathcal{B}_t$ is a basis of \black{the cycle space} $C(K_t)$ and that the dimension of $C(K_t)$ is $|\mathcal{B}_t|=\binom{t-1}{2}$.

Denote by $e_C$ the smallest $\{i,j\}\in \binom{[t]}{2}$ in the lexicographic order such that $(i,j)\in C$ or $(j,i)\in C$.
Next, we will prove the claim by a reverse induction on the lexicographic order of $e_C$. Note that $t\not\in e_C$ since $|C|\geq 3$, which implies that the claim is vacuously true when $e_C=\{t-1,t\}$ (which never occurs).
Now assume that the claim holds for all oriented cycles $C'$ with $e_{C'}> e_{C}$.
Let $\{i,j\}=e_C$, where $i<j$.
We may assume that $(i,j)\in C$ by flipping the orientation of $C$ if necessary, which corresponds to negating  $u^C$.

Let $s$ be the number of directed edges  that $C$ and $\Delta_{ijt}$ share. If $s=3$ then it is clear by definition that
$C=\Delta_{ijt}$, and we are done. Otherwise, $1\leq s\leq 2$ and it is easy to verify \black{by a case analysis} that $u^C-u^{\Delta_{ijt}}=u^{C'}$ for a set $C'$ that is either an oriented cycle in $G$ or a disjoint union of two oriented cycles $C_1, C_2$ in $G$ passing through $t$. The latter case occurs when \black{$s=1$, $C$ passes through $t$, and $\{(t,i),(j,t)\}\cap C=\emptyset$}. \black{As $C'$ is obtained from $C$ by removing its smallest edge $(i,j)$ (under the lexicographic order) and adding one (when $s=2$) or two (when $s=1$) edges greater $(i,j)$,} in either case, the smallest edge of $C'$  is greater than the edge $e_C=\{i,j\}$. Hence, by the induction hypothesis and the fact that $u^{C'}=u^{C_1}+u^{C_2}$ when $C'$ is the disjoint union of $C_1$ and $C_2$, we have $u^{C'}\in W$. So $u^C=u^{C'}+u^{\Delta_{ijt}}\in W$, completing the proof of the claim.
\end{proof}

The  basis $\mathcal{B}_t$ is also viewed as a $\binom{t-1}{2}\times \binom{t}{2}$ matrix over $\mathbb{F}$ whose columns are labeled
by the edges $\{i,j\}$ of $K_t$, according to the lexicographic order defined above. Moreover, the rows of $\mathcal{B}_t$ represent $u^{\Delta_{ijt}}$ for $1\le i<j\le t-1$, and are labeled by $\{i,j\}\in \binom{[t-1]}{2}$, also according to the lexicographic order. 
For example, $\mathcal{B}_3=(1,-1,1)$ and
\begin{equation}\label{Def-B_4}
  \begin{aligned}
    \mathcal{B}_4=\left(
      \begin{array}{cccccc}
        1 &  &  & -1 & 1 &  \\
         & 1 &  & -1 &  & 1 \\
         &  & 1 &  & -1 & 1 \\
      \end{array}
    \right),
  \end{aligned}
\end{equation}
where the $6$ columns are labeled and ordered lexicographically by  $\{1,2\}<\{1,3\}<\{2,3\}<\{1,4\}<\{2,4\}<\{3,4\}$.
Observe for example that the $\pm 1$ entries in the first row correspond to the oriented triangle $\Delta_{124}=\{(1, 2), (2,4), (4,1)\}$, where we have $-1$ on the column labeled by the edge $\{1,4\}$, since the directed edge $(4,1)$ in $\Delta_{124}$ has the opposite orientation from the  orientation of the edge in \black{the oriented $K_t$}. 

We remark that the above definition of $\mathcal{B}_t$, is  given with respect to the fixed orientation of the edges of \black{the oriented $K_t$},
as with the definition of $u^C$ for any oriented cycle $C$. One may define $\mathcal{B}_t$ with respect to other orientations of edges, which corresponds to changing the signs in some columns. These definitions are all equivalent and the analysis in this paper holds for any orientation up to change of signs.

Moreover, when the characteristic of $\mathbb{F}$ is two,  we recover the definition of $\mathcal{B}_t$  in \cite{shangguan2019combinatorial0} using the fact that $1=-1$. While working in the case $\mathrm{char}(\mathbb{F})=2$ has the advantage that  there is no need to distinguish the signs, the theory holds more generally over any field.

\subsection{$t$-Wise Intersection Matrices}

We proceed to define $t$-wise intersection matrices, but we begin with a few preliminary definitions.
Given  $n$  variables or field elements $x_1,\ldots,x_n$, define the  $n\times k$ Vandermonde matrix
\begin{equation}\label{vandermonde}
V_k(x_1,\ldots,x_n)=\left(
    \begin{array}{cccc}
      1  & x_{1} & \cdots & x_{1}^{k-1} \\
       &  & \ddots &  \\
      1  & x_{n} & \cdots & x_{n}^{k-1} \\
    \end{array}
  \right).
\end{equation}
When the $x_i$'s are understood from the context, for $I\subseteq [n]$, we use the abbreviation $V_k(I):=V_k(x_i:i\in I)$ to denote the restriction of $V_k(x_1,\ldots,x_n)$ to the rows with indices in $I$.

Let $\mathcal{I}_k$ denote the identity matrix of order $k$. Next, we give the definition of $t$-wise intersection matrices.
\begin{definition}[$t$-wise intersection matrices]\label{Def-matrix-general}
For a positive integer $k$ and $t\geq 3$ subsets $I_1,\ldots,I_t\subseteq[n]$, the $t$-wise intersection matrix  $M_{k,(I_1,\ldots,I_t)}$  is  the
$(\binom{t-1}{2}k+\sum_{1\le i<j\le t}|I_i\cap I_j|)\times \binom{t}{2}k$ variable matrix with entries in $\mathbb{F}[x_1, \ldots,x_n]$, defined as

\begin{equation*}\label{int-matrix}
\black{M_{k,(I_1,\ldots,I_t)}=}\left(
  \begin{array}{c}
    \mathcal{B}_t\otimes \mathcal{I}_k  \\\hline

  {\rm diag}\Big(V_k(I_i\cap I_j): \{i,j\}\in \binom{[t]}{2}\Big) \\
  \end{array}
\right),
\end{equation*}
where $\otimes$ is tensor product of matrices and
\begin{itemize}
\item $\mathcal{B}_t\otimes \mathcal{I}_k$ is a $\binom{t-1}{2}k\times \binom {t}{2}k$ matrix with entries in $\{0, \pm 1\}$ \black{(recall that the matrix $\mathcal{B}_t$ is defined below the proof of Lemma \ref{cycle-space-basis})},

\item ${\rm diag}\big(V_k(I_i\cap I_j): \{i,j\}\in \binom{[t]}{2}\big)$ is a block diagonal matrix with blocks  $V_k(I_i\cap I_j)$,  ordered by the lexicographic order on $\{i,j\}\in\binom{[t]}{2}$. Note that this matrix has order $(\sum_{1\le i<j\le t}|I_i\cap I_j|)\times \binom {t}{2}k$. If $I_i\cap I_j=\emptyset$ for some $i,j$, then $V_k(I_i\cap I_j)$ is of order $0\times k$ and the $\{i,j\}\in \binom{[t]}{2}$ block of $k$ columns is a $\sum_{1\le i<j\le t}|I_i\cap I_j|\times k$ zero matrix.
\end{itemize}
\end{definition}

\noindent \black{To avoid possible confusion, we want to emphasize that we use the capital $I_1,\ldots,I_t$ to denote subsets of $[n]$, and use the curly $\mathcal{I}_k$ to denote the identity matrix of order $k$.} The reader is referred to Appendix~\ref{app:1} (see Example \ref{example-1}) for an example of a $4$-wise intersection matrix. We note that when $t=2$, $\mathcal{B}_t$ is an empty matrix and $M_{k,(I_1,I_2)}$ is simply a Vandermonde matrix.

For a vector $\alpha\in\mathbb{F}^n$, the evaluation of $M_{k,(I_1,\ldots,I_t)}$ at the vector $\alpha$ is denoted by $M_{k,(I_1,\ldots,I_t)}(\alpha)$, where each variable $x_i$ is assigned the value $\alpha_i$. Given subsets $I_1,\ldots,I_t\subseteq[n]$, we call the variable matrix $M_{k,(I_1,\ldots,I_t)}$ {\it nonsingular} if it contains at least one $\binom{t}{2}k\times\binom{t}{2}k$ submatrix whose determinant is a nonzero polynomial in $\mathbb{F}[x_1,\ldots,x_n]$.

The paper \cite{shangguan2019combinatorial0} connects the nonsingularity of intersection matrices to the list-decodability of RS codes.
We will use this connection to prove our main result, Theorem~\ref{thm:main}.

However, we will first prove that \em certain \em intersection matrices are nonsingular.  This will both allow us to cleanly illustrate the connection to disjoint tree packings of graphs, and it will also yield Theorem~\ref{perfect-hash-matrix} about perfect hash matrices.  We will do this in Theorem~\ref{theorem-weaker} in the next section.

\section{Connection to Tree Packing and an Intermediate Result}
\label{sec:3}

In this section we prove the following theorem.  We recall from \eqref{definition-weight} the definition of the \em weight \em of a collection of sets:

\begin{equation}
  {\rm wt}(I_1,\ldots,I_t)=\sum_{i=1}^t |I_i|-\left|\bigcup_{i=1}^t I_i\right|.
\end{equation}

\begin{theorem}\label{theorem-weaker}
  Let $t\ge 2$ be an integer and $I_1,\ldots,I_t\s[n]$ be subsets satisfying \rm{(i)} $I_i\cap I_j\cap I_{\ell}=\emptyset$ for all $1\le i<j<\ell\le t$; \rm{(ii)} $\wt(I_J)\leq (|J|-1)k$ for all  nonempty $J\subseteq [t]$; \rm{(iii)}  $\wt(I_{[t]})=(t-1)k$. Then the $t$-wise intersection matrix $M_{k,(I_1,\ldots,I_t)}$ is nonsingular over any field.
\end{theorem}

\black{Recall that for $J\subseteq [t]$, we use the shorthand   ${\rm wt}(I_J):={\rm wt}(I_j:j\in J)$.}
As discussed above, this theorem stops short of Conjecture~\ref{conjecture}, due to the assumption that $I_i \cap I_j \cap I_\ell = \emptyset$.  In the language of list-decoding Reed--Solomon codes, this only gives us a statement about lists of potential codewords that have no three-wise intersections.
However, we will build on this statement to prove our main theorem about list-recovery (Theorem~\ref{thm:main}), and moreover this is already enough to
prove our result on the existence of strongly perfect hash matrices (Theorem \ref{perfect-hash-matrix}).

The main tool of proving Theorem \ref{theorem-weaker} is the following classical result in graph theory.

\begin{lemma}[Nash-Williams \cite{nash1961edge}, Tutte \cite{tutte1961problem}, see also Theorem 2.4.1 of \cite{diestel2017graphentheory}]\label{Nash-Williams}
  A multigraph contains $k$ edge-disjoint spanning trees if and only if for every partition $\mathcal{P}$ of its vertex set it has at least $(|\mathcal{P}|-1)k$ cross-edges.
  Here an edge is called a cross-edge for $\mathcal{P}$ if its two endpoints are in different members of $\mathcal{P}$.
\end{lemma}

In order to apply the Nash-Williams--Tutte theorem, we will construct a graph $G$ from the sets $I_1, I_2, \ldots, I_t$.  We first note that the assumptions on $I_1, \ldots, I_t$ from Theorem~\ref{theorem-weaker} imply some nice properties that will later allow us to apply Lemma~\ref{Nash-Williams}.

\begin{claim}\label{cl:size}
Suppose that $I_1,\ldots,I_t$ are subsets satisfying the assumptions of Theorem \ref{theorem-weaker}.
Then the  matrix $M_{k,(I_1,\ldots,I_t)}$ is a square matrix of order $\binom{t}{2}k$. Further, for any $J\subseteq [t]$ with $|J|\ge 2$,
\begin{equation}\label{constraint-weaker}
  {\rm wt}(I_J)=\sum_{\{i,j\}\in\binom{J}{2}}|I_i\cap I_j|.
\end{equation}
\end{claim}

\begin{proof}
By \eqref{definition-weight} (the definition of weight) and the inclusion-exclusion principle \begin{equation}
\label{stam}
{\wt}(I_{[t]})=\sum_{i=1}^t|I_i|-\left|\bigcup_{i=1}^tI_i\right|=\sum_{j=2}^t\sum_{J\in\binom{[t]}{j}}(-1)^{|J|}\left|\bigcap_{i\in J}I_i\right|.
\end{equation} Therefore, by assumption (i) of Theorem \ref{theorem-weaker} we have $\wt(I_{[t]})=\sum_{1\le i<j\le t} |I_i\cap I_j|$. Then, by assumption (iii)
the  matrix $M_{k,(I_1,\ldots,I_t)}$ is in fact a square matrix of order $\binom{t}{2}k$. Similarly, \eqref{constraint-weaker} holds for any $J\subseteq [t]$ with $|J|\ge 2$.
\end{proof}

To prove Theorem \ref{theorem-weaker}, let us construct a multigraph $G$ defined on a set $V$ of $t$ vertices, say $V=\{v_1,\ldots,v_t\}$. For $1\le i<j\le t$, connect vertices $v_i,v_j$ by $|I_i\cap I_j|$ multiple edges.

Applying Lemma \ref{Nash-Williams} to $G$ leads to the following claim.

\begin{claim}\label{claim-spanning-tree}
Let $G$ be as above.  Then
$G$ contains $k$ edge-disjoint spanning trees.
\end{claim}

\begin{proof}
Let $\mathcal{P}=\{V_1,\ldots,V_s\}$ be an arbitrary partition of $V$. Then it is clear that $\sum_{i=1}^s |V_i|=t$. According to Lemma \ref{Nash-Williams}, to prove the claim it suffices to show that $G$ has at least $(s-1)k$ cross-edges with respect to $\mathcal{P}$. By \eqref{constraint-weaker} and assumption (iii) of Theorem \ref{theorem-weaker} it is easy to see that $G$ contains $\sum_{1\le i<j\le t}|I_i\cap I_j|=(t-1)k$ edges. Moreover, by \eqref{constraint-weaker} and assumption (ii) of Theorem \ref{theorem-weaker} one can infer that for each $i\in[s]$, the induced subgraph of $G$ on the vertex set $V_i$ has at most $\wt(I_j:j\in V_i)\le (|V_i|-1)k$ edges. It follows that the number of cross-edges of $G$ (with respect to $\mathcal{P}$) is at least
$$
(t-1)k-\sum_{i=1}^s (|V_i|-1)k=\left(t-1-\sum_{i=1}^s |V_i|+s\right)k=(s-1)k,
$$
as needed, thereby completing the proof of the claim.
\end{proof}

Below, we will relate a tree packing of this graph $G$ to the determinant of the intersection matrix $M_{k,(I_1, \ldots, I_t)}$.  In order to do this, we first record
a property of the matrix $\mathcal{B}_t$.  Recall that the columns of $\mathcal{B}_t$ are indexed by ${[t] \choose 2}$.

\begin{claim}\label{claim-property-of-B_t}
If we remove a set of columns from $\mathcal{B}_t$, then the new matrix will have the same row rank as $\mathcal{B}_t$ if and only if the columns are labeled by an acyclic subgraph of $K_t$.
\end{claim}

\begin{proof}
First we prove the if direction.
Assume to the contrary that we can remove from  $\mathcal{B}_t$ some columns labeled by an acyclic subgraph $H$ of $K_t$ and reduce the row rank.
Let $\mathcal{B}'_t$ be the submatrix of $\mathcal{B}_t$ after the removal of the columns labeled by  $H$. The rows of   $\mathcal{B}'_t$ are linearly dependent by assumption. Hence, there exists a nonzero vector $u\in\mathbb{F}^{\binom{t-1}{2}}$ such that $u\cdot\mathcal{B}'_t=0$. As $u\neq 0$ and the rows of $\mathcal{B}_t$ are linearly independent, we have $u\cdot\mathcal{B}_t\neq 0$. Let $S\subseteq \binom{[t]}{2}$ be the support of $u\cdot\mathcal{B}_t$, where the support of a vector of length $n$ is the subset of $[n]$ that records the indices of its nonzero coordinates. As $u\cdot\mathcal{B}_t\neq 0$ and $u\cdot\mathcal{B}'_t=0$, we have $\emptyset\neq S\subseteq H$.

Consider the  $\binom{t}{2}\times t$ matrix $D=(D_{\{i,j\},s})$
which is defined   by
\[
D_{\{i,j\},s}=\begin{cases}
1 & s=j, \\
-1 & s=i, \\
0 & \text{otherwise},
\end{cases}
\]
where  $1\leq i<j\leq t \text{ and } s\in [t]$. Note that the rows and columns of $D$ are  labeled by $\{i,j\}\in \binom{[t]}{2}$
and  $s\in [t]$ respectively. It is not hard to check that $\mathcal{B}_t\cdot D=0$. \black{Indeed, for $1\le i<j\le t-1$ and $1\le s\le t$, let us compute the inner product of the $\{i,j\}$-th row of $\mathcal{B}_t$, i.e. $u^{\Delta_{ijt}}$, and the $s$-th column of $D$, denoted by $D^s$. As $u^{\Delta_{ijt}}$ has only three nonzero coordinates: $u^{\Delta_{ijt}}_{\{i,j\}}=1,~u^{\Delta_{ijt}}_{\{i,t\}}=-1$, and $u^{\Delta_{ijt}}_{\{j,t\}}=1$, by the definition of $D$ one can observe that
\begin{itemize}
    \item [(i)] if $s\not\in\{i,j,t\}$ then $D^s_{\{i,j\}}=D^s_{\{i,t\}}=D^s_{\{j,t\}}=0$;
    \item [(ii)] if $s=i$ then $D^s_{\{i,j\}}=D^s_{\{i,t\}}=-1,~D^s_{\{j,t\}}=0$;
    \item [(iii)]  if $s=j$ then $D^s_{\{i,j\}}=1,~D^s_{\{i,t\}}=0,~D^s_{\{j,t\}}=-1$;
    \item [(iv)] if $s=t$ then $D^s_{\{i,j\}}=0,~D^s_{\{i,t\}}=D^s_{\{j,t\}}=1$.
\end{itemize}
One can check that in all of the four cases listed above, the inner product $u^{\Delta_{ijt}}\cdot D^s=0$, which implies that $\mathcal{B}_t\cdot D=0$ and hence $u\cdot \mathcal{B}_t\cdot D=0$.}

Denote $u\cdot \mathcal{B}_t$ by $w=(w_{\{i,j\}})\in\mathbb{F}^{\binom{t}{2}}$, whose support is $S$.
As  $\emptyset\neq S\subseteq H$ and $H$ is acyclic, we can find $s_0\in [t]$ whose degree in $S$ is one, i.e., there exists a unique edge $\{i_0,j_0\}\in S$ such that $s_0\in \{i_0,j_0\}$. Then, the $s_0$-th entry of $w\cdot D$ is
\[
\sum_{\{i,j\}\in \binom{[t]}{2}} w_{\{i,j\}} D_{\{i,j\},s_0} = w_{\{i_0,j_0\}} D_{\{i_0,j_0\},s_0}=\pm  w_{\{i_0,j_0\}}\neq 0,
\]
which is a contradiction as $w\cdot D= u\cdot \mathcal{B}_t\cdot D=0$.

Now we prove the only if direction.
It suffices to prove that removing from $\mathcal{B}_t$ a set of columns labeled by a cycle $C$ of $K_t$ will reduce its row rank by at least $1$. Let us orient the edges of $C$ to make it an oriented cycle, which by abuse of notation is also denoted by $C$.
Since the rows of $\mathcal{B}_t$ form a basis of \black{the cycle space} $C(K_t)$, there is a nonzero vector $u\in\mathbb{F}^{\binom{t-1}{2}}$ such that $u\cdot\mathcal{B}_t=u^C$. \black{Note that $u^C$ is a vector of length $\binom{t}{2}$ with support being exactly $C$, i.e., for $\{i,j\}\in\binom{[t]}{2}$, $u^C_{\{i,j\}}\neq 0$ if and only if $\{i,j\}\in C$.} Let $\mathcal{B}'_t$ be the submatrix of $\mathcal{B}_t$ after the removal of the columns labeled by $C$. Then it is not hard to check that $u\cdot\mathcal{B}'_t=0$, which implies that the rows of $\mathcal{B}'_t$ are linearly dependent, as needed.
\end{proof}

Next we present the proof of Theorem \ref{theorem-weaker}. Recall from Claim~\ref{cl:size}
 that under the assumptions of Theorem \ref{theorem-weaker}, the $t$-wise intersection matrix
\begin{equation*}\label{int-matrix-2}
M_{k,(I_1,\ldots,I_t)}=\left(
  \begin{array}{c}
    \mathcal{B}_t\otimes \mathcal{I}_k  \\\hline

  {\rm diag}\Big(V_k(I_i\cap I_{j}): \{i,j\}\in \binom{[t]}{2}\Big) \\
  \end{array}
\right),
\end{equation*}

\noindent is a square matrix of order $\binom{t}{2}k$, and is defined by exactly $(t-1)k$ variables $x_s,~s\in S$, where $S\subseteq[n]$ is some subset of size $(t-1)k$. In order to prove that $M_{k,(I_1,\ldots,I_t)}$ is nonsingular, we proceed to show the nonsingularity of the following matrix, obtained by permuting the columns and rows of $M_{k,(I_1,\ldots,I_t)}$:
\begin{equation*}
M'_{k,(I_1,\ldots,I_t)}:=\left(
  \begin{array}{c}
    \mathcal{I}_k\otimes \mathcal{B}_t  \\\hline

  \big(\ma{C}_i:0\le i\le k-1\big) \\
  \end{array}
\right),
\end{equation*}

\noindent where $\ma{C}_i={\rm diag}\Big(V_k^{(i)}(I_j\cap I_{j'}): \{j,j'\}\in \binom{[t]}{2}\Big)$ and $V_k^{(i)}(I_j\cap I_{j'})$ is the $(i+1)$-th column of $V_k(I_j\cap I_{j'})$. \black{More precisely, one can write $\mathcal{C}_i={\rm diag}\Big((x_{\ell}^i:\ell\in I_j\cap I_{j'})^T:\{j,j'\}\in \binom{[t]}{2}\Big)$.}
Above,
$\big(\ma{C}_i:0\le i\le k-1\big)$ is a $(t-1)k\times\binom{t}{2}k$ variable matrix,
which consists of the matrices $\ma{C}_i$ stacked next to each other.
See Figure~\ref{fig:mat2} for an illustration, and Example~\ref{example-2} in Appendix~\ref{app:1} for a concrete example.

\begin{figure}
\centering
\footnotesize
\begin{tikzpicture}[scale=.5]
\begin{scope}
\draw[thick] (0,2) rectangle (12, 14);
\draw[thick] (0,8) -- (12, 8);
\foreach \j in {8, 10, 12}
{
\draw[thin] (0,\j) -- (12, \j);
}
\foreach \j in {2,4,6,8,10}
{
\draw[thin] (\j,8) -- (\j, 14);
}
\draw[fill=black!10] (0,7) rectangle (2,8);
\draw[fill=black!10] (2,6) rectangle (4, 7);
\draw[fill=black!10] (4,5) rectangle (6, 6);
\draw[fill=black!10] (6,4) rectangle (8, 5);
\draw[fill=black!10] (8,3) rectangle (10, 4);
\draw[fill=black!10] (10,2) rectangle (12, 3);

\node[black](a) at (5.6, 16) {\begin{minipage}{5cm}\begin{center} This column indexed by $\{j,\ell\} \in {[t]\choose 2}$ and $i \in [k]$\end{center}\end{minipage}};
\node[black](b) at (5,18) {\begin{minipage}{4cm}\begin{center} This row is indexed by $x_s$, for $s \in I_j \cap I_\ell$.\end{center}\end{minipage}};
\draw[->,black](a) to (5.6, 14);
\draw[->,black](b) to [out=180,in=160] (0,5.5);
\draw[dashed,black](0,5.5) to (5.6,5.5);
\draw[dashed,black](5.6,14) to (5.6, 5.5);

\node[black](c) at (2,4) {\begin{minipage}{5cm}\begin{center} This is $x_s^i$\end{center}\end{minipage}};
\draw[fill=black] (5.6, 5.5) circle (.2cm);
\draw[->,black](c) to (5.6, 5.5);

\node at (6,1) {$M_{k, (I_1, \ldots, I_t)}$};

\foreach \j in {(1, 13), (3,11), (5,9), (9,13), (11,11), (11,9)}{
\node at \j {$\mathcal{I}_k$};
}
\foreach \j in {(7, 13), (7,11), (9,9)}{
\node at \j {$-\mathcal{I}_k$};
}
\end{scope}
\begin{scope}[xshift=14cm]
\draw[thick] (0,2) rectangle (12, 14);
\draw[thick] (0,8) -- (12, 8);
\foreach \j in {8, 10, 12}
{
\draw[thin] (0,\j) -- (12, \j);
}
\foreach \j in {4,8}
{
\draw[thin] (\j,8) -- (\j, 14);
}
\draw[fill=black!10] (0,7) rectangle (.66,8);
\draw[fill=black!10] (4,7) rectangle (4.66,8);
\draw[fill=black!10] (8,7) rectangle (8.66,8);
\draw[fill=black!10] (0.66,6) rectangle (1.33, 7);
\draw[fill=black!10] (4.66,6) rectangle (5.33, 7);
\draw[fill=black!10] (8.66,6) rectangle (9.33, 7);
\draw[fill=black!10] (1.33,5) rectangle (2, 6);
\draw[fill=black!10] (5.33,5) rectangle (6, 6);
\draw[fill=black!10] (9.33,5) rectangle (10, 6);
\draw[fill=black!10] (2,4) rectangle (2.66, 5);
\draw[fill=black!10] (6,4) rectangle (6.66, 5);
\draw[fill=black!10] (10,4) rectangle (10.66, 5);
\draw[fill=black!10] (2.66,3) rectangle (3.33, 4);
\draw[fill=black!10] (6.66,3) rectangle (7.33, 4);
\draw[fill=black!10] (10.66,3) rectangle (11.33, 4);
\draw[fill=black!10] (3.33,2) rectangle (4, 3);
\draw[fill=black!10] (7.33,2) rectangle (8, 3);
\draw[fill=black!10] (11.33,2) rectangle (12, 3);


\foreach \j in {(2, 13), (6,11), (10,9)}{
\node at \j {$\mathcal{B}_t$};
}

\node[black](a) at (9.66, 16) {\begin{minipage}{5cm}\begin{center} This column indexed by $i \in [k]$ and $\{j,\ell\} \in {[t] \choose 2}$ \end{center}\end{minipage}};
\node[black](b) at (8.66,18) {\begin{minipage}{4cm}\begin{center} This row is indexed by $x_s$, for $s \in I_j \cap I_\ell$.\end{center}\end{minipage}};
\draw[->,black](a) to (9.66, 14);
\draw[->,black](b) to [out=0,in=0] (12,5.5);
\draw[dashed,black](12,5.5) to (9.66,5.5);
\draw[dashed,black](9.66,14) to (9.66, 5.5);

\node[black](c) at (8.5,4) { $x_s^i$};
\draw[fill=black] (9.66, 5.5) circle (.2cm);
\draw[->,black](c) to (9.66, 5.5);
\node at (6,1) {$M'_{k, (I_1, \ldots, I_t)}$};
\end{scope}

\end{tikzpicture}

\caption{
Re-ordering the rows/columns of an intersection matrix.
(In this cartoon, $t=4$ and $k=3$).
}
\label{fig:mat2}
\end{figure}

\begin{proof}[\textbf{Proof of Theorem \ref{theorem-weaker}}]
If $t=2$, then $M_{k,(I_1,I_2)}$ is a $k\times k$ Vandermonde matrix, which is nonsingular, so assume $t\ge 3$.
For the rest of the proof, we will consider the matrix $M' = M'_{k,(I_1, \ldots, I_t)}$ discussed above, and show that it is nonsingular.

Let the graph $G$ be as in the discussion above;
recall that for distinct $i,j\in [t]$, two vertices $v_i,v_{j}$ in $G$ are connected by $|I_i\cap I_{j}|$ edges.
By Claim~\ref{cl:size}, $M'$ is a square matrix with ${t \choose 2}k$ rows and columns, and $k(t-1)$ ``variable'' rows at the bottom.  Let $S\subseteq [n]$ be the subset that records the indices of variables $x_s$ that appear in $M'$.

For $1\le i<j\le t$, fix an arbitrary one-to-one correspondence between the $|I_i \cap I_j|$ edges connecting $v_i,v_{j}$ and the $|I_i \cap I_j|$  variables $x_s \in S$ so that $s\in I_i\cap I_{j}$.
Since any three distinct subsets $I_i,I_{j},I_{\ell}$ have empty intersection, this yields a one-to-one correspondence $$\phi:E(G)\longrightarrow \{x_s:s\in S\},$$
between the $(t-1)k$ edges of $G$ and the $(t-1)k$ variables with indices in $S$.

By Claim \ref{claim-spanning-tree}, the edges of $G$ can be partitioned into $k$ edge-disjoint spanning trees $T_i$, and $G=\bigcup_{i=0}^{k-1} T_i$.

Observe that for each $0\le i\le k-1$, $\mathcal{C}_i$ has entries that are either zero or of the form $x_s^i$ for some $x_s \in S$.
We will show how to use the tree decomposition of $G$ to choose nonzero entries in each $\mathcal{C}_i$ so that (a) every row in the bottom part of $M'$ is chosen exactly once, and (b) when the columns chosen are removed from $M'$, the resulting submatrix of $\mathcal{B}_t$ is nonsingular.  This will mean that the product of these non-zero entries appears in the determinant expansion of $M'$.

For each $i$, we pick $t-1$ non-zero elements from each $\mathcal{C}_i$: we choose $x_s^i$ for  $x_s \in \{\phi(e) : e \in T_i\}$.   That is, we consider all of the variables $x_s$ corresponding to edges that appear in $T_i$.
Let $m_i(x)$ denote the product of these entries:
\[ m_i(x) = \prod_{x_s \in \{\phi(e)~:~ e \in T_i\}} x_s^i. \]
Let $m(x)=\prod_{i=0}^{k-1} m_i(x)$. Since $\phi$ is a bijection, $m(x)$ is a product of $(t-1)k$ distinct entries chosen from the submatrix $\big(\ma{C}_i:0\le i\le k-1\big)$, and crucially, no two of them appear in the same row or column \black{--- no two in the same row because, by definition, the trees $T_i$ are pairwise edge-disjoint, and no two in the same column because a tree by definition has no parallel edges.}

To conclude the proof, it is enough to show that $m(x)$ appears as a nonvanishing term in the determinant expansion of $M'_{k,(I_1,\ldots,I_t)}$. Indeed, removing from $M'_{k,(I_1,\ldots,I_t)}$ the $(t-1)k$ rows and columns that correspond to $m(x)$, the resulting submatrix is a block diagonal matrix
\[
{\rm diag}\Big(\mathcal{B}'_t(i):~0\le i\le k-1\Big),
\]
where for each $i$, $\mathcal{B}'_t(i)$ is a square submatrix of $\mathcal{B}_t$ of order $\binom{t-1}{2}$. By construction, each $\mathcal{B}'_t(i)$ is obtained by removing from $\mathcal{B}_t$ a set of $t-1$ columns labeled by the spanning tree $T_i$.  By Claim \ref{claim-property-of-B_t}, this implies that $\mathcal{B}'_t(i)$ is nonsingular.
Moreover, as each of the sets $I_i \cap I_j$ are disjoint due to the assumptions of the theorem, the monomial $m(x)$ appears only once in the determinant expansion of $M'_{k,(I_1, \ldots, I_t)}$.
Consequently, the ``coefficient'' of $m(x)$ in the determinant expansion of $M'_{k,(I_1,\ldots,I_t)}$ is nonvanishing, completing the proof of the theorem.
\end{proof}

\section{Application to Perfect Hashing}
\label{phf}

In this section, we apply Theorem~\ref{theorem-weaker} to perfect hashing, and we prove Theorem~\ref{perfect-hash-matrix}.
\black{Recall that given a matrix and a set $S$ of its columns, a row is said to {\it separate} $S$ if, restricted to this row, these columns have distinct values. For a positive integer $t$, a matrix is said to be a {\it $t$-perfect hash matrix} if any set of $t$ distinct columns of the matrix is separated by at least one row. }

\begin{theorem*}[Theorem~\ref{perfect-hash-matrix}, restated]
Given integers $1\le k<n$ and $t\ge 3$, for a sufficiently large prime power $q$ \black{(as an \black{exponential} function of $n$ and $t$)}, there exists an $n\times q^k$ matrix, defined on the alphabet $\mathbb{F}_q$, such that any set of $t$ columns is separated by at least $n-k(t-1)+1$ rows.
\end{theorem*}

We will also show that Theorem~\ref{perfect-hash-matrix} is optimal, at least within the class of \em linear \em hash matrices.
Generalizing a definition of \cite{Blackburn1998Optimal} (with a slightly different terminology), we say that an $n\times q^k$ matrix
$M$ is called {\it linear} if it is defined over the  field $\mathbb{F}_q$ and has the form $M=PQ$, where $P$ is an $n\times k$ coefficient matrix and $Q$ is the $k\times q^k$ matrix whose columns are formed by the $q^k$ distinct vectors of $\mathbb{F}_q^k$.

With this terminology, we will prove the following proposition, which generalizes a result of \cite{Blackburn1998Optimal} (see Theorem 4 of \cite{Blackburn1998Optimal}).
\begin{proposition}\label{proposition}
  If a linear $n\times q^k$ matrix  separates any  set of $t$ columns by at least $r$ rows, then $r\leq n-k(t-1)+1$.
\end{proposition}
Proposition~\ref{proposition} implies that the bound in Theorem~\ref{perfect-hash-matrix} is tight, at least for linear constructions. \black{In fact, according to a question of Blackburn and Wild (see Question 4 in \cite{Blackburn1998Optimal}), for fixed integers $k,n,t$ and sufficiently large prime power $q$, the construction given by Theorem~\ref{perfect-hash-matrix} is likely to be optimal among all strongly perfect hash matrices (not necessarily linear). In fact, using our notation, Blackburn and Wild suspected that Proposition~\ref{proposition} holds also after removing the linearity condition.}

\subsection{Proof of Theorem~\ref{perfect-hash-matrix}}
Fix $\mathbb{F}$ to be the finite field $\mathbb{F}_q$, and let us begin with an overview of the proof.
Recall that an evaluation vector  of $\mathbb{F}_q^n$ is a vector whose coordinates are all distinct.
It is well-known that an $[n,k]$-RS code over  $\mathbb{F}_q^n$  is of size $q^k$, and that  any two distinct codewords agree on at most  $k-1$  coordinates. For our purpose, we view an $[n,k]$-RS code as an $n\times q^k$ matrix whose  columns are the codewords of the code. More precisely, the columns are all  the vectors $$\big\{\big(f(\alpha_1),\ldots,f(\alpha_n)\big)^T,~f\in\mathbb{F}_q[x],~\deg(f)< k\big\}$$ with some arbitrary ordering, and   $\alpha=(\alpha_1,\ldots,\alpha_n)$ is the evaluation vector that defines the code. We say that the evaluation vector $\alpha$ {\it defines} the  $n\times q^k$ matrix.

Fix an integer $t\ge 3$. An evaluation vector $\alpha\in\mathbb{F}_q^n$ is called ``bad'' if it does not define a strongly $t$-perfect hashing matrix.
The main idea in the proof of Theorem \ref{perfect-hash-matrix} is to show that the number of bad evaluation vectors is at most $O_{n,t}(q^{n-1})$, whereas   there are $\frac{q!}{(q-n)!}=\Theta_n(q^n)$ distinct evaluation vectors. Therefore,  for sufficiently large $q$ there must exist an evaluation vector that is not bad, i.e., it defines an $n\times q^k$ strongly $t$-perfect hash matrix.

The main tool used in proving the upper bound on the number of bad evaluation vectors is the following well-known result.

\begin{lemma}[DeMillo-Lipton-Schwartz-Zippel lemma, see, e.g., \cite{jukna2011extremal} Lemma 16.3]\label{zeros}
A nonzero polynomial $f\in\mathbb{F}_q[x_1,\ldots,x_n]$ of degree $d$ has at most $dq^{n-1}$ zeros  in $\mathbb{F}_q^n$.
\end{lemma}

We need three more lemmas before presenting the proof of Theorem \ref{perfect-hash-matrix}.

\begin{lemma}\label{main-lemma}
  Given integers $1\le k<n,~t\ge 3$, if  an evaluation vector $\alpha\in\mathbb{F}_q^n$ does not define an    $n\times q^k$ strongly $t$-perfect hashing matrix, then there exists an integer $s\in\{3,\ldots,t\}$ and subsets $I_1,\ldots,I_{s}\subseteq [n]$ such that
  \begin{itemize}
 \item [{\rm (i)}] $I_i\cap I_j\cap I_{\ell}=\emptyset$ for all $1\le i<j<\ell\le s$;
 \item [{\rm (ii)}] ${\rm wt}(I_{J})\leq (|J|-1)k$ for any nonempty subset $J\subseteq [s]$;
 \item [{\rm (iii)}] ${\rm wt}(I_{[s]})=(s-1)k$;
 \item [{\rm (iv)}] the $s$-wise intersection matrix $M_{k,(I_1,\ldots,I_{s})}$ over $\mathbb{F}_q$ is a nonsingular square matrix of order $k\binom{s}{2}$, whose determinant is a nonzero polynomial in $\mathbb{F}_q[x_1,\ldots,x_n]$ of degree less than $k^2t$.
\end{itemize}
\end{lemma}

\begin{proof}
If an evaluation vector $\alpha\in\mathbb{F}_q^n$ is bad, then the $n\times q^k$ matrix it defines contains $t$ distinct columns defined by polynomials $f_1,\ldots,f_t$, which are  separated by at most $n-k(t-1)$ rows. Equivalently, there are at least  $k(t-1)$ rows that do not separate these $t$ columns.

Next, we iteratively construct  the sets $I_1,\ldots,I_t\subseteq [n]$. We set all of them to be the empty set, and then for each row $i$ that does not separate the $t$ columns, we add $i$ to arbitrary two sets $I_j,I_{\ell}$ for which $f_j(\alpha_i)=f_{\ell}(\alpha_i)$. It is easy to verify that the sets $I_j$ satisfy the following properties
\begin{itemize}
 \item [{\rm (a)}] $I_i\cap I_j\cap I_{\ell}=\emptyset$ for all $1\le i<j<\ell\le t$;
 \item [{\rm (b)}] $|I_i\cap I_j|\le k-1$ for distinct $i,j\in[t]$;
 \item [{\rm (c)}] ${\rm wt}(I_{[t]})=\sum_{1\le i<j\le t} |I_i\cap I_j|\geq k(t-1)$.
\end{itemize}

\noindent Indeed,  {\rm (a)}  follows from the definition of $I_1,\ldots,I_t$, {\rm (b)} follows from the property of RS codes, and {\rm (c)} follows from \eqref{stam} and {\rm (a)} (the equality part), \black{and the assumption that there are at least $k(t-1)$ rows that do not separate the $t$ columns (the inequality part)}.

Let \black{$s\in\{3,\ldots,t\}$} be the smallest positive integer  for which there exist a subset $S\subseteq [t]$ of size $s$ with  ${\rm wt}(I_{S})\ge k(s-1)>0$. \black{Then, by {\rm (b)} $s\ge 3$ and by {\rm (c)} $s\le t$. Therefore,} $s$ is well-defined. 
 Assume without loss of generality that $S=[s]$.

By construction and the minimality of $s$, the sets $I_1,\ldots, I_s$ satisfy properties {\rm (i)} and {\rm (ii)}. We proceed to verify that also  {\rm (iii)} holds. Note that properties {\rm (i)} and {\rm (ii)} continue to hold if one removes an element from one of the sets $I_j$, and by doing so,  the  weight ${\rm wt}(I_{[s]})$
    can reduce by at most one.   Hence, by iteratively removing elements from the sets $I_j$, one can construct sets, which we also denote by $I_1,\ldots, I_s$, that satisfy property {\rm (iii)}, while retaining properties {\rm (i)} and {\rm (ii)}.

Since the subsets $I_1,\ldots,I_{s}\subseteq [n]$ satisfy the  three assumptions of Theorem \ref{theorem-weaker}, it holds that  $M_{k,(I_1,\ldots,I_{s})}$ is a nonsingular matrix. \black{Moreover, it follows by Claim \ref{cl:size} that $M_{k,(I_1,\ldots,I_{s})}$ is a square matrix of order $k\binom{s}{2}$. It remains to prove the claimed upper bound on the degree of the determinant. As $M_{k,(I_1,\ldots,I_{s})}$ is a square matrix of order $k\binom{s}{2}$, by Definition \ref{Def-matrix-general} the lower (variable) part ${\rm diag}\big(V_k(I_i\cap I_j): \{i,j\}\in \binom{[s]}{2}\big)$ consists of exactly $(s-1)k$ rows. According to \eqref{vandermonde}, every entry of ${\rm diag}\big(V_k(I_i\cap I_j): \{i,j\}\in \binom{[s]}{2}\big)$ is a monomial of degree at most $k-1$. Therefore, the total degree of the determinant cannot exceed $(s-1)k(k-1)<k^2s\le kt$, as needed,}
thereby completing the proof of {\rm (i)}-{\rm (iv)}.
\end{proof}

\begin{lemma}\label{lem:required}
    Let $s\ge 3$ be an integer and $f_1,\ldots,f_s\in\mathbb{F}_q^k$. Let  $f=(f_i-f_j: 1\leq i<j\leq s)\in\mathbb{F}_q^{\binom{s}{2}k}$, which is the concatenation of the vectors  $f_i-f_j$  according to the lexicographic order on $\binom{[s]}{2}$ defined in Section \ref{preliminaries}. Then it follows that \begin{equation*}
 (\mathcal{B}_s\otimes \mathcal{I}_k)\cdot f^T=0.
\end{equation*}
\end{lemma}

\begin{proof}
    Note that for any $1\le i<j\le s$, we have
\begin{equation}\label{triangle}
(f_i-f_j)+(f_j-f_s)-(f_i-f_s)=0.
\end{equation}
Recall that the row vectors of $\mathcal{B}_{s}$ correspond to the oriented triangles $\Delta_{ijs}$. Then it follows from \eqref{triangle} and the definition of $\mathcal{B}_{s}$ that
\begin{equation*}
 (\mathcal{B}_s\otimes \mathcal{I}_k)\cdot f^T=0,
\end{equation*}
as needed.
\end{proof}

\begin{lemma}\label{main-lemma-2}
  Let $\alpha$ be an evaluation vector that
  does not define an $n\times q^k$ strongly $t$-perfect hashing matrix, and let
  $M_{k,(I_1,\ldots,I_{s})}$  be the $s$-wise intersection matrix for $3\leq s\leq t$ given by Lemma \ref{main-lemma}. Then,  the matrix $M_{k,(I_1,\ldots,I_{s})}(\alpha)$, which is the evaluation of  $M_{k,(I_1,\ldots,I_{s})}$ at $\alpha$,  does not have full rank.
\end{lemma}

\begin{proof}
To prove the lemma, it suffices to show that the matrix  $M_{k,(I_1,\ldots,I_{s})}(\alpha)$ has a nontrivial kernel. Towards this end, let $f_1,\ldots,f_{s}$ be the $s$  distinct polynomials that correspond to the set $S=[s]$ found in the proof of  Lemma \ref{main-lemma}. \black{Note that the polynomials $f_1,\ldots,f_s$ correspond to subsets $I_1,\ldots,I_s$ that satisfy assumptions (i)-(iv) of Lemma \ref{main-lemma}}.

We view a polynomial of degree at most $k-1$ also as a   vector of length $k$  defined by its $k$ coefficients, where for $1\le i\le k$, the $i$-th coordinate of the vector is the coefficient of the monomial $x^{i-1}$ in that polynomial.  
Let  $f=(f_i-f_j: 1\leq i<j\leq s)\in\mathbb{F}_q^{\binom{s}{2}k}$, as defined in the statement of Lemma \ref{lem:required}. 

We claim that $M_{k,(I_1,\ldots,I_{s})}(\alpha)\cdot f^T=0$. First of all, it follows by Lemma \ref{lem:required} that
\begin{equation}\label{first-half}
 (\mathcal{B}_s\otimes \mathcal{I}_k)\cdot f^T=0.
\end{equation}

\noindent Moreover, observe that by definition, for any $\ell\in I_i\cap I_j$ we have $f_i(\alpha_{\ell})=f_j(\alpha_{\ell})$, which implies that $0=f_i(\alpha_{\ell})-f_j(\alpha_{\ell})=(f_i-f_j)(\alpha_{\ell})=(f_i-f_j)(\alpha_{\ell})$. Therefore  $V_k(I_i\cap I_j)\cdot f_{ij}^T=0$, which implies that
\begin{equation}\label{second-half}
 {\rm diag}\Big(V_k(I_i\cap I_j): \{i,j\}\in \binom{[s]}{2}\Big)\cdot f^T=0.
\end{equation}

\noindent Combining \eqref{first-half} and \eqref{second-half} we conclude that $M_{k,(I_1,\ldots,I_{s})}(\alpha)\cdot f^T=0$, completing the proof of the claim.
\end{proof}

We are now in a position to prove Theorem \ref{perfect-hash-matrix}.

\begin{proof}[\textbf{Proof of Theorem \ref{perfect-hash-matrix}}]
We give an upper bound on the number of bad evaluation vectors that do not define a strongly $t$-perfect hash matrix.

Let $\mathcal{M}$ be the set of  $s$-wise intersection matrix $M_{k,(I_1,\ldots,I_{s})}$ that satisfy conditions {\rm (i)}-{\rm (iv)} of Lemma \ref{main-lemma} for any $s\in \{3,\ldots, t\}$.
It is clear that any  $s$-wise intersection matrix $M_{k,(I_1,\ldots,I_s)}$ is completely determined by the subsets $I_1,\ldots,I_s$, therefore the size of $\mathcal{M}$    is at most  $2^{nt}$.

 By Lemma \ref{main-lemma}, for any bad evaluation vector $\alpha\in\mathbb{F}_q^n$ there exists a matrix $M\in\mathcal{M}$ whose determinant is a nonzero polynomial in $\mathbb{F}_q[x_1,\ldots,x_n]$ of degree less than $k^2t$, \black{and moreover},  by Lemma \ref{main-lemma-2}, it holds that $\det(M)(\alpha)=0$. Therefore, the set of bad evaluation vectors is contained in the union of the zero sets of the polynomials $\det(M),~M\in\mathcal{M}$, which by Lemma \ref{zeros}, is of size at most
$$|\mathcal{M}|\cdot(k^2t)q^{n-1}\le (2^{nt}k^2t)q^{n-1}=O_{n,t}(q^{n-1}).$$
The result follows by observing that the number of evaluation vectors, i.e., the number of  vectors in $\mathbb{F}_q^n$ with pairwise distinct coordinates is  $\frac{q!}{(q-n)!}=\Theta_n(q^n)$. Hence, for sufficiently large $q$ there exist many evaluation vectors that are not bad, and the result follows.
\end{proof}

\subsection{Proof of Proposition~\ref{proposition}}
Next, we prove Proposition~\ref{proposition}, which implies that Theorem~\ref{perfect-hash-matrix} is tight, at least for linear hash matrices.
\begin{proof}[\textbf{Proof of Proposition \ref{proposition}}]
It is clear that any $n\times q^k$ matrix $M$ given by an $[n,k]$-RS code is linear, as we may write $M=V_k(\alpha_1,\ldots,\alpha_n)\cdot Q$, where $\alpha=(\alpha_1,\ldots,\alpha_n)$ is the evaluation vector,  $V_k(\alpha_1,\ldots,\alpha_n)$ is the associated $n\times k$ Vandermonde matrix as defined in \eqref{vandermonde}, and $Q$ is the $k\times q^k$ matrix whose columns are formed by the $q^k$ distinct vectors of $\mathbb{F}_q^k$.
The statement on the optimality follows  from  Theorem 4 of \cite{Blackburn1998Optimal} which claims that  any linear $N\times q^k$ matrix that has the property that any set of $t$ columns is separated by at least one row, satisfies $N\ge k(t-1)$.

By this result, the restriction to the first   $k(t-1)-1$ rows of an $n\times q^k$ linear matrix, contains a set of $t$ columns that are not separated at all. Hence, this set of columns is separated by at most $n-k(t-1)+1$ rows of the $n\times q^k$ linear matrix, and the result follows.
\end{proof}

\section{Near-Optimal List-Recovery of RS Codes: Proof of the Main Theorem}
\label{list-recovery}

In this section, we prove our main theorem on the list-recovery of RS codes, Theorem~\ref{thm:main}.
For $n,k\in\mathbb{N}^+$, a prime power $q>1$, and $\alpha_1,\dots,\alpha_n\in\mathbb{F}_q$, define
\[
\mathsf{RS}_{n,k,q}(\alpha_1,\dots,\alpha_n):=\{(f(\alpha_1),\dots,f(\alpha_n)) : f(x)\in \mathbb{F}_q[x], ~\deg(f)<k\}\subseteq\mathbb{F}_q^n.
\]
In particular, when $\alpha_1,\dots,\alpha_n$ are distinct, $\mathsf{RS}_{n,k,q}(\alpha_1,\dots,\alpha_n)$ is the $[n,k]$ RS code over $\mathbb{F}_q$ with the evaluation points $\alpha_1,\dots,\alpha_n$.

We will in fact prove the following theorem, which implies Theorem~\ref{thm:main}.
\begin{theorem} \label{list-recover-main}
Let $k,n, L, \ell\in\N^+$, $\epsilon\in (0,1]$, and $\delta>0$ be such that  $L\geq (1+\delta)\ell/\epsilon-1$  and
\[
k/n\leq \frac{\epsilon}{c \sqrt{\ell}( \frac{1+\delta}{\delta } )(\log(\frac{1}{\epsilon})+\log( \frac{1+\delta}{\delta })+1)},
\]
where $c>0$ is the constant in Lemma~\ref{lem_subset}.
Suppose $q\geq 2^{c'(L+n\log L)}$ for a large enough constant $c'>0$ and $\alpha_1,\dots,\alpha_n$ are chosen uniformly and independently from $\mathbb{F}_q$ at random.
Then with high probability \black{($\ge 1-O_{n,L}(1/q)$)}, the code $\mathsf{RS}_{n,k,q}(\alpha_1,\dots,\alpha_n)$ has rate $R=k/n$ and is list-recoverable up to relative distance $1-\epsilon$ with input list size $\ell$ and output list size $L$.
In particular, by choosing $\delta$ to be any positive constant, we could achieve $L=O(\ell/\epsilon)$ and $R=\Omega\left(\frac{\epsilon}{\sqrt{\ell}(\log(1/\epsilon)+1)}\right)$.
\end{theorem}

We begin with an overview of the proof.

\subsection{Overview of the Proof}
\label{sec:overview}

We give an overview of our proof of Theorem~\ref{list-recover-main}.
For simplicity, let us first assume the input list size $\ell$ equals one, i.e., we restrict to the case of list decoding.
In this case, Theorem~\ref{list-recover-main} states that there exist RS codes of rate $\Omega(\frac{\epsilon}{\log(1/\epsilon)+1})$ that are list-decodable from radius $1-\epsilon$ with list size $O(1/\epsilon)$.

As discussed previously, Conjecture~\ref{conjecture} about the nonsingularity of intersection matrices would be enough to establish Theorem~\ref{list-recover-main}, and indeed an even stronger result.
While we do not know if Conjecture~\ref{conjecture} holds in general,  Theorem~\ref{theorem-weaker} states that it holds under an extra condition that $I_i\cap I_{i'}\cap I_{i''}=\emptyset$ for distinct $i,i',i''\in[t]$. Our proof of Theorem~\ref{list-recover-main} is based on this theorem.

As Theorem~\ref{theorem-weaker} requires the above extra condition,  which does not hold in general, we cannot simply follow the proof in \cite{shangguan2019combinatorial0} and replace  Conjecture~\ref{conjecture} by Theorem~\ref{theorem-weaker}.
One naive way of fixing this  is removing elements from the sets $I_i$ until the condition $I_i\cap I_{i'}\cap I_{i''}=\emptyset$ for distinct $i,i',i''\in[t]$ is satisfied. Specifically, for each $j\in [n]$ such that there exist more than two sets $I_{i_1}, \dots, I_{i_s}$ containing $j$, we pick two sets (say $I_{i_1}$ and $I_{i_2}$) and remove $j$ from all the other sets. The resulting sets $I'_1,\dots,I'_t$ satisfy the condition $I'_i\cap I'_{i'}\cap I'_{i''}=\emptyset$ for distinct $i,i',i''\in[t]$ and we can now apply Theorem~\ref{theorem-weaker} to conclude that $M_{k,(I'_1,\ldots,I'_t)}$ is nonsingular.

The problem with this idea, however, is that  $\wt(I'_{[t]})$ could be much smaller than $\wt(I_{[t]})$.
Indeed, if a coordinate $j\in [n]$ is simultaneously contained in $s\geq 2$ sets $I_{i_1},\dots,I_{i_s}$, then it is counted $s-1$ times in   $\wt(I_{[t]})=\sum_{i=1}^t |I_i|-\left|\bigcup_{i=1}^t I_i\right|$ (one for each set $I_{i_u}$, $u=1,\dots,s$, and minus one for $\bigcup_{i=1}^t I_i$). On the other hand, this coordinate is counted only once in $\wt(I'_{[t]})$ since we have removed it from all but two of the sets $I_{i_u}$.
As $s$ can be as large as $t$, the new weight $\wt(I'_{[t]})$ can be smaller than $\wt(I_{[t]})$ by a factor of $t-1=\Omega(1/\epsilon)$.
As a consequence, implementing this idea directly only yields RS codes of rate $\Omega(\epsilon^2)$.

To mitigate this problem, we perform a random sampling of the collection $\{I_1,\dots,I_t\}$ before removing elements from $I_i$. Namely, we choose a random subset $J\subseteq [t]$ of some appropriate cardinality to be determined later.
Then,  we remove elements from the sets $I_i$ just like before, but only for $i\in J$, so that the resulting sets $I'_i$ satisfy the condition $I'_i\cap I'_{i'}\cap I'_{i''}=\emptyset$ for distinct $i,i',i''\in J$.
Finally, we apply Theorem~\ref{theorem-weaker} to conclude that  the $|J|$-wise intersection matrix $M_{k,(I'_i)_{i\in J}}$ is nonsingular, which can still be used to prove the list-decodability of the RS code.

The advantage of replacing $[t]$ by the random sample $J\subseteq [t]$ is that the condition $\wt(I'_{[t]})\geq (t-1)k$
is replaced by $\wt(I'_{J})\geq (|J|-1)k$. It turns out that the condition $I_i'\cap I_{i'}'\cap I_{i''}'=\emptyset$ for distinct $i,i',i''\in J$ is easier to satisfy since  $|J|$ may be much smaller than $t$.
Consequently, we are able to show that there exist RS codes of rate $\Omega(\frac{\epsilon}{\log(1/\epsilon)+1})$ using this improved method.

Finally, we explain how to choose the cardinality of the sample $J$.
Let $j\in [n]$ and denote by $s_j$ the number of sets among $I_1,\dots, I_t$ that contain  $j$.
Then for the index $j$, we choose $|J|=\Theta(t/s_j)$ for our purpose.
However, the number $s_j$ may vary  when $j$ ranges over $[n]$, meaning that there may not be a single choice of $|J|$ that works best for all $j\in [n]$ simultaneously.

We solve this problem using the following trick: Create a logarithmic number of ``buckets'' and put $j\in [n]$ in the $i$-th bucket if $2^{i-1} \leq s_j<2^i$. Then choose $|J|$ according to the heaviest bucket. Here, we lose a factor of $O(\log(1/\epsilon)+1)$ in the rate because there are about $\log L = O(\log(1/\epsilon)+1)$ buckets.

\paragraph{Generalization to list recovery.} In the case of list decoding, we choose each set $I_i$ to be the subset of coordinates where a codeword $c_i$ and the received word $y$ agree.  In the more general setting of list recovery, each $I_i$ is the subset of $j\in [n]$ where $j$-th coordinate of $c_i$ belongs to a certain set $S_j\subseteq \mathbb{F}_q$ of size $\ell$.
So the $j$-th coordinate of $c_i$ may take at most $\ell$ possible values, which may be thought of as $\ell$ different ``colors."

One way of extending our proof to list recovery is choosing subsets $I'_1\subseteq I_1,\dots,I'_t\subseteq I_t$ that come from coordinates of codewords with the same color, where the colors are picked such that $\wt(I'_{[t]})=\wt(I'_1,\dots, I'_t)$ is maximized. Then we proceed as in the case of list decoding.
It is not hard to show that this yields RS codes of rate $\Omega(\frac{\epsilon}{\ell(\log(1/\epsilon)+1)})$ which are list-recoverable from radius $1-\epsilon$ with input list size $\ell$ and output list size $O(\ell/\epsilon)$.

With a more careful analysis, we show that we can achieve a better rate  $\Omega(\frac{\epsilon}{\sqrt{\ell}(\log(1/\epsilon)+1)})$, as stated by Theorem~\ref{thm:main}.
Our analysis is inspired by \cite{LP20} which proved a similar result on the list-recoverability of randomly punctured
codes with a different setting of parameters.

\subsection{A Combinatorial Lemma}

In this subsection, we state a combinatorial lemma (Lemma~\ref{lem_subset}). It guarantees the existence of a subset $J\subseteq [t]$ and sets $I'_i\subseteq I_i$ for $i\in J$  that satisfy certain conditions, particularly the condition  $I'_{i}\cap I'_{i'}\cap I'_{i''}=\emptyset$ for distinct  $i,i',i''\in J$. We then use this lemma together with Theorem~\ref{theorem-weaker} to  prove Theorem~\ref{thm:main}.
The proof of this combinatorial lemma is postponed to Subsection~\ref{sec_comb}.

First, we need the following generalization of the weight function $\mathrm{wt}(\cdot)$.

\begin{definition}[Generalized weight function]
Let $n,t\in\mathbb{N}^+$ and $I_1,\dots, I_t\subseteq [n]$. Let $U_j=\{i\in [t]: j\in I_i\}$ for $j\in [n]$. For $J\subseteq [t]$ and $\ell\in\mathbb{N}^+$, define the \emph{$\ell$-th generalized weight} $\mathrm{wt}_{\ell}(I_J)$ of $I_J$ to be
\[
\mathrm{wt}_\ell(I_J):=\sum_{j=1}^n \max\{|U_j\cap J|-\ell, 0\}.
\]
\end{definition}
Note that $\mathrm{wt}(I_J)=\sum_{i\in J} |I_i|-|\bigcup_{i\in J} I_i|=\mathrm{wt}_1(I_J)$. Also note that
\begin{equation}\label{eq_weightell}
\mathrm{wt}_\ell(I_J)\geq\sum_{j=1}^n  \left(|U_j\cap J|-\ell\right) = \sum_{i\in J} |I_i| - \ell n.
\end{equation}

The proof of  Theorem~\ref{thm:main} uses the following combinatorial lemma, which we prove in the next subsection.

\begin{lemma}\label{lem_subset}
Let $k,n,t,\ell\in\mathbb{N}^+$, $\epsilon\in (0,1]$, $\delta>0$, and $I_1, \dots, I_t\subseteq [n]$. Let $S_1,\dots, S_n$ be sets of size $\ell$ with the following property: for each $j\in[n]$, $\{c_{i,j}:j\in I_i\}\subseteq S_j$, where the $c_{i,j}$'s are field elements.
Suppose $t\geq (1+\delta)\ell/\epsilon$,
$\left|I_i\right|\geq \epsilon n$ for $i\in [t]$, and
\[
  \mathrm{wt}_\ell(I_{[t]})\geq \left(c\sqrt{\ell}\left(\log\left(\frac{1}{\epsilon}\right)+\log\left( \frac{1+\delta}{\delta }\right)+1\right)\right)\cdot tk.
\]
where $c>0$ is a large enough absolute constant. Then there exist $J\subseteq [t]$ and a collection $(I_i')_{i\in J}$ of subsets of $[n]$ indexed by $J$
such that $|J|\geq 2$, $I_i'\subseteq I_i$ for $i\in J$, and the following conditions are satisfied:
\begin{enumerate}[(1)]
\item $I'_{i}\cap I'_{i'}\cap I'_{i''}=\emptyset$ for distinct  $i,i',i''\in J$.
\item $\mathrm{wt}(I'_{J'})\leq (|J'|-1)k$ for all nonempty $J'\subseteq J$.
\item $\mathrm{wt}(I'_J)=(|J|-1)k$.
\item $c_{i,j}=c_{i',j}$ for $i,i'\in [t]$ and $j\in I'_i\cap I'_{i'}$.
\end{enumerate}
\end{lemma}

\begin{remark}
 Condition~(4) is introduced for  list recovery.
The case $\ell=1$ corresponds to list decoding.
In this case, Condition~(4) is automatically satisfied.
\end{remark}

 We also need the following lemma that bounds the number of  pairs $(J, (I_i')_{i\in J})$.

\begin{lemma}\label{lem_counting}
The number of  $(J, (I_i')_{i\in J})$ satisfying Condition~(1) of Lemma~\ref{lem_subset} is at most $2^t(1+t+\binom{t}{2})^n$.
\end{lemma}
\begin{proof}
There are at most $2^t$ choices of $J$.
Now fix $J\subseteq [t]$.
For $j\in [n]$, let $T_j=\{i\in J: j\in I'_i\}$.
Note that we have $|T_j|\leq 2$ for all $j\in [n]$ by Condition~(1) of Lemma~\ref{lem_subset}.
So for each $j\in [n]$, the number of choices of $T_j$ is at most $1+t+\binom{t}{2}$.
Also note that the sets $I_i'$ are determined by the sets $T_j$ by $I_i'=\{j\in [n]: i\in T_j\}$.
So the number of choices of $(J, (I_i')_{i\in J})$ is at most $2^t(1+t+\binom{t}{2})^n$.
\end{proof}

\subsection{Proof of Theorem~\ref{list-recover-main}}

Now we are ready to prove Theorem~\ref{list-recover-main}, the main theorem of Section~\ref{list-recovery}.
%
%

\begin{proof}[Proof of Theorem~\ref{list-recover-main}]
At the high level, the proof works as follows: First, we show that with high probability, a certain square matrix $M_{k, (I_i')_{i\in J}}(\alpha_1,\dots,\alpha_n)$ has full rank. Next, we show that whenever the matrix has full rank,  $\mathsf{RS}_{n,k,q}(\alpha_1,\dots,\alpha_n)$ is list-recoverable up to relative distance $1-\epsilon$ as claimed by Theorem~\ref{list-recover-main}. More precisely, we prove the contrapositive, i.e., if $\mathsf{RS}_{n,k,q}(\alpha_1,\dots,\alpha_n)$ is not list-recoverable up to relative distance $1-\epsilon$, then there exists a nonzero vector $f$ in the kernel of $M_{k, (I_i')_{i\in J}}(\alpha_1,\dots,\alpha_n)$. The proof that such a vector $f$ exists uses Lemma~\ref{lem_subset}, Lemma \ref{lem:required}, and the same ideas in the proof of Lemma~\ref{main-lemma-2}.

Let $t=L+1$. Consider the following two conditions:
\begin{enumerate}[(1)]
\item $\alpha_i\neq \alpha_j$ for all distinct $i,j\in [n]$.
\item For all $J\subseteq [t]$ and $(I_i')_{i\in J}$ satisfying Conditions~(1)--(3) of Lemma~\ref{lem_subset}, we have
\[
\det(M_{k, (I_i')_{i\in J}}(\alpha_1,\dots,\alpha_n))\neq 0,
\] where $M_{k, (I_i')_{i\in J}}$ denotes the  $\binom{|J|}{2}k\times \binom{|J|}{2}k$ variable matrix\footnote{The number of rows of $M_{k, (I_i')_{i\in J}}$ is $\binom{|J|-1}{2}k + \sum_{\{i,j\}\in \binom{J}{2}} |I'_i\cap I'_j|$, which equals $\binom{|J|-1}{2}k+\sum_{i\in J} |I'_i|-|\bigcup_{i\in J} I'_i|=\binom{|J|-1}{2}k+\wt(I'_J)$ by Condition~(1) of Lemma~\ref{lem_subset}. This number further equals $\binom{|J|-1}{2}k+(|J|-1)k=\binom{|J|}{2}k$ by Condition~(3) of Lemma~\ref{lem_subset}.}
\[
M_{k, (I_i')_{i\in J}}=\left(
  \begin{array}{c}
    \mathcal{B}_{|J|}\otimes \mathcal{I}_k  \\\hline
  {\rm diag}\Big(V_k(I_i'\cap I_j'): \{i,j\}\in \binom{J}{2}\Big) \\
  \end{array}
\right).
\]
\end{enumerate}
The first condition is satisfied with probability at least $1-\binom{n}{2}/q$.
For the second condition, consider fixed $J\subseteq [t]$ and $(I_i')_{i\in J}$  satisfying Conditions~(1)--(3) of Lemma~\ref{lem_subset}.
We  know $\det(M_{k, (I_i')_{i\in J}})\neq 0$ by Theorem~\ref{theorem-weaker}.
 Also note that $\det(M_{k, (I_i')_{i\in J}})$ is a multivariate polynomial of total degree at most $(|J|-1)k(k-1)\leq Lk^2$.
So by Lemma \ref{zeros}, $\det(M_{k, (I_i')_{i\in J}})(\alpha_1,\dots,\alpha_n)\neq 0$ holds with probability at least
$1-Lk^2/q$ for  fixed $J$ and $(I_i')_{i\in J}$. The number of choices of $(J, (I_i')_{i\in J})$ is at most  $2^t(1+t+\binom{t}{2})^n$ by Lemma~\ref{lem_counting}. By the union bound, the two conditions are simultaneously satisfied with probability at least
\[
1-\binom{n}{2}/q-2^t\left(1+t+\binom{t}{2}\right)^n Lk^2/q=1-o(1)
\]
over the random choices of $\alpha_1,\dots,\alpha_n$, where we use the assumption that $q\geq 2^{c'(L+n\log L)}$ and $c'>0$ is a large enough constant.

Fix $\alpha_1,\dots,\alpha_n\in\F_q$ that satisfy the above two conditions.
By the first condition, the code $\mathsf{RS}_{n,k,q}(\alpha_1,\dots,\alpha_n)$ has rate exactly $k/n$.
It remains to show that $\mathsf{RS}_{n,k,q}(\alpha_1,\dots,\alpha_n)$ is  list-recoverable up to relative distance $1-\epsilon$ with input list size $\ell$ and output list size $L$.
Assume to the contrary that this does not hold. Then there exist $t$ distinct polynomials $f_1,\dots,f_t\in \F_q[x]$ of degree less than $k$ and sets $S_1,\dots,S_n$ of size $\ell$ such that for all $i\in [t]$, the cardinality of the set
\[
I_i:=\{j\in [n]:  f_i(\alpha_j)\in S_j\}
\]
is at least $\epsilon n$.

As $t=L+1\geq (1+\delta) \ell/\epsilon$ and
$k/n\leq \frac{\epsilon}{c \sqrt{\ell}( \frac{1+\delta}{\delta } )(\log(\frac{1}{\epsilon})+\log( \frac{1+\delta}{\delta })+1)}$,
we have
\[
\mathrm{wt}_\ell(I_{[t]})\stackrel{\eqref{eq_weightell}}{\geq } t \epsilon n-\ell n \geq \frac{\delta}{1+\delta} \cdot t\epsilon n\geq  \left(c\sqrt{\ell}\left(\log\left(\frac{1}{\epsilon}\right)+\log\left( \frac{1+\delta}{\delta }\right)+1\right)\right)\cdot tk.
\]
By Lemma~\ref{lem_subset},  there exist $J\subseteq [t]$ and   $(I_i')_{i\in J}$
such that $|J|\geq 2$, $I_i'\subseteq I_i$ for $i\in J$,  and Conditions~(1)--(4) of Lemma~\ref{lem_subset} are satisfied, where we let $c_{i,j}=f_i(\alpha_j)\in S_j$ for $i\in [t]$ and $j\in I_t$.

Let $f=(f_i-f_j: \{i,j\}\in \binom{J}{2}, i<j)\in \F_q^{\binom{|J|}{2}k}$, as defined in the statement of Lemma \ref{lem:required}. As $|J|\geq 2$ and $f_1,\dots, f_t$ are distinct, we have $f\neq 0$.

We claim that
\begin{equation}\label{eq_lineqns}
M_{k, (I_i')_{i\in J}}(\alpha_1,\dots,\alpha_n) \cdot f^T=0.
\end{equation}
To see this, first note that it follows by Lemma \ref{lem:required} that $(\mathcal{B}_{|J|}\otimes \mathcal{I}_k)\cdot f^T=0$.
Now consider a row $v$ of the submatrix
\[
{\rm diag}\Big(V_k(I_i'\cap I_j'): \{i,j\}\in \binom{J}{2}\Big)(\alpha_1,\dots,\alpha_n),
\]
 of $M_{k, (I_i')_{i\in J}}(\alpha_1,\dots,\alpha_n)$, which corresponds to some $\{i,j\}\in \binom{J}{2}$ with $i<j$ and $s\in I_i'\cap I_j'$. By definition, we have $v \cdot f^T=(f_i-f_j)(\alpha_s)=f_i(\alpha_s)-f_j(\alpha_s)$, i.e., the row $v$ represents the linear constraint $f_i(\alpha_s)-f_j(\alpha_s)=0$.
By Condition~(4) of Lemma~\ref{lem_subset}, we have $f_i(\alpha_s)=c_{i,s}=c_{j,s}=f_j(\alpha_s)$.
So $v\cdot f^T=f_i(\alpha_s)-f_j(\alpha_s)=0$.
This proves \eqref{eq_lineqns}.

Finally, by \eqref{eq_lineqns}, we have $\det(M_{k, (I_i')_{i\in J}}(\alpha_1,\dots,\alpha_n))=0$. But this contradicts the choice of $\alpha_1,\dots,\alpha_n$.
\end{proof}

\subsection{Proof of Lemma~\ref{lem_subset}}\label{sec_comb}

We present the proof of Lemma~\ref{lem_subset} in this subsection.

Let $k,n,t,\ell\in\N^+$, $\epsilon\in (0,1]$, $\delta>0$, $I_1,\dots,I_t$,
$S_1,\dots, S_n$,
and $c_{i,j}$ for $i\in [t]$ and $j\in I_i$
be as in Lemma~\ref{lem_subset}.
In particular, we have $t\geq (1+\delta)\ell/\epsilon$,
 $|I_i|\geq \epsilon n$ for $i\in [t]$, and
\begin{equation}\label{eq_weightbound}
  \mathrm{wt}_\ell(I_{[t]})\geq \left(c\sqrt{\ell}\left(\log\left(\frac{1}{\epsilon}\right)+\log\left( \frac{1+\delta}{\delta }\right)+1\right)\right)\cdot tk.
\end{equation}
where $c>0$ is a large enough absolute constant.
We will choose $J\subseteq [t]$ to be a random set and  carefully construct the sets $I'_i$, $i\in J$, such that the conditions in Lemma~\ref{lem_subset} are satisfied.

For $j\in [n]$, let $U_j:=\{i\in [t]: j\in I_i\}$. By definition, we have
\begin{equation}\label{eq_weight_ell}
\mathrm{wt}_\ell(I_{[t]})=\sum_{j=1}^n \max\{|U_j|-\ell, 0\}.
\end{equation}
Assume for a moment that there exists an integer $K\in\N^+$ such that $\max\{|U_j|-\ell, 0\}$ equals either $K$ or zero for all $j\in [n]$. Then by \eqref{eq_weight_ell}, the number of $j\in [n]$ for which $\max\{|U_j|-\ell, 0\}=K$ holds (or equivalently, $|U_j|=K+\ell$ holds) is precisely $\mathrm{wt}(I_{[t]})/K$.
The next lemma extends this fact to the general case with only logarithmic loss.

\begin{lemma}\label{lem_card}
There exists an integer $K>0$ such that
 the number of $j\in [n]$ satisfying $|U_j|\geq K+\ell$
  is at least $\frac{\mathrm{wt}_\ell(I_{[t]})}{d_0 K (\log(\frac{1}{\epsilon})+\log( \frac{1+\delta}{\delta })+1)}$,
 where $d_0>0$ is some absolute constant.
\end{lemma}
\begin{proof}
By \eqref{eq_weightell} and the fact that $t\geq (1+\delta)\ell /\epsilon$,  we have
\begin{equation}\label{eq_weight_lb}
\mathrm{wt}_\ell(I_{[t]})\geq t \epsilon n-\ell n\geq \frac{\delta}{1+\delta} \cdot t\epsilon n.
\end{equation}
For $i=0,1,2,\dots$, let $B_i=\{j\in [n]: 2^i\leq |U_j|-\ell <2^{i+1}\}$.
Then
\[
\mathrm{wt}_\ell(I_{[t]})=\sum_{j=1}^n \max\{|U_j|-\ell, 0\}=\sum_{i=0}^{\lceil\log t\rceil-1} \sum_{j\in B_i}  (|U_j|-\ell).
\]
Let $d=\lfloor \log ( \frac{\delta}{1+\delta} \cdot t\epsilon/2) \rfloor$. Note that $d$ could be negative (possibly $\frac{\delta}{1+\delta} \cdot t\epsilon/2\in(0,1)$).
Then
\[
\sum_{0\leq i<d} \sum_{j\in B_i}  (|U_j|-\ell)\leq n 2^d \leq  \frac{\delta}{1+\delta} \cdot t\epsilon n/2  \stackrel{\eqref{eq_weight_lb}}{\leq} \mathrm{wt}_\ell(I_{[t]})/2.
\]
Therefore
\begin{equation}\label{eq_truncate}
\sum_{i=\max\{d,0\}}^{\lceil\log t\rceil-1} \sum_{j\in B_i} (|U_j|-\ell) \geq \mathrm{wt}_\ell(I_{[t]})/2.
\end{equation}

Let $\Delta=\lceil\log t\rceil-\max\{d,0\}= O(\log(\frac{1}{\epsilon})+\log( \frac{1+\delta}{\delta })+1)$.
By \eqref{eq_truncate}, there exists an integer $i_0$ such that $\max\{d,0\}\leq i_0\leq \lceil\log t\rceil-1$ and
\begin{equation}\label{eq_i0}
 \sum_{j\in B_{i_0}}  (|U_j|-\ell)\geq \frac{\mathrm{wt}_\ell(I_{[t]})}{2\Delta}.
\end{equation}
  Choose $K=2^{i_0}$. Then $K\leq |U_j|-\ell<2K$ for all $j\in B_{i_0}$.
The upper bound $  |U_j|-\ell <2K$  for  $j\in B_{i_0}$, together with \eqref{eq_i0},    implies $|B_{i_0}|\geq  \frac{\mathrm{wt}_\ell(I_{[t]})}{4K \Delta}$.
So the number of $j\in [n]$ satisfying  $|U_j|\geq K+\ell$ is at least $\frac{\mathrm{wt}_\ell(I_{[t]})}{4K\Delta}=\Omega\left(\frac{\mathrm{wt}_\ell(I_{[t]})}{K(\log(\frac{1}{\epsilon})+\log( \frac{1+\delta}{\delta })+1)}\right)$.
\end{proof}

Fix $K$ satisfying Lemma~\ref{lem_card}. Define
\[
A:=\{j\in [n] : |U_j|\geq K+\ell\}\subseteq [n].
\]
By the choice of $K$ and Lemma~\ref{lem_card}, we have
\begin{equation}\label{eq_boundA}
|A|\geq \frac{\mathrm{wt}_\ell(I_{[t]})}{d_0 K(\log(\frac{1}{\epsilon})+\log( \frac{1+\delta}{\delta })+1)}.
\end{equation}


We also need the following technical lemma.

  \begin{lemma}
    For real numbers $p\in(0,\frac{1}{2}]$ and $x\ge 0$, we have $(1-p)^x(1+px)\le 1-\frac{1}{8}p^2x^2$ if $x\le \frac{1}{p}$, and $(1-p)^x(1+px)\le \frac{2}{e}$ otherwise.
  \label{lem:calc}
  \end{lemma}
  \begin{proof}
    Fix $p$ and let $f(y) = (1-p)^y(1+py)$.
    For $y\geq 0$, the derivative $f'(y)$ satisfies
    \begin{equation} \label{eq:calc}
    \begin{aligned}
    f'(y) & =  (1-p)^y(\ln(1-p)\cdot (1+py) + p)\\ & \le (1-p)^y(-p(1+py)+p)\\
    & = -p^2y(1-p)^y.
    \end{aligned}
    \end{equation}
    So $f'(y) \le -p^2y(1-p)^x$ for $y\in [0,x]$.
    As $f(0)=1$, we have $f(x)\le 1+\int_{0}^{x} -p^2y(1-p)^x dy = 1-\frac{1}{2}p^2x^2(1-p)^x$.
    If $x\le \frac{1}{p}$, we have $(1-p)^x\ge (1-p)^{1/p}\ge 1/4$, as $(1-p)^{1/p}$ is decreasing with $p$, and thus is minimized at $p=\frac{1}{2}$.
    Hence, we have $f(x)\le 1-\frac{1}{8}p^2x^2$.
    By \eqref{eq:calc}, $f(y)$ is decreasing and thus maximized at $y=1/p$ on the interval $[1/p,\infty)$, so for $x\ge 1/p$, we have $f(x)\le f(1/p) =(1-p)^{1/p}(1+1) \le \frac{2}{e}$.
  \end{proof}

The above lemma is used to prove the following statement.

 \begin{lemma}\label{lem_sampling}
   Choose a random subset $J\subseteq [t]$ by independently including each $i\in[t]$ in $J$ with probability $p=\min\{\sqrt{\ell}/(2K),\frac{1}{2}\}$.
    Let
    \[
    A_J:=\{j\in A: \text{ there exist distinct } i, i'\in J \text{ such that } j\in I_i\cap I_{i'} \text{ and } c_{i,j}=c_{i',j}\}.
    \]
    Then $\ex[|A_J|]=\Omega(|A|)$.
 \label{}
 \end{lemma}
 \begin{proof}
  Fix $j\in A$.
  It suffices to prove that $\Pr[j\in A_J]\ge c$ for some constant $c$. Suppose $S_j=\{c_1,\dots,c_\ell\}$. For $r=1,\dots,\ell$, let
  \[
   U_j^{(r)}:=\{i\in U_j: c_{i,j}=c_r\}
   \quad\text{and}\quad t_r\coloneqq \max\{|U_j\ind{r}|-1,0\}.
   \]
   Then $U_j=\bigcup_{r=1}^\ell U_j^{(r)}$.
   Let $K'\coloneqq \sum_{r=1}^{\ell} t_j$.
   Since $j\in A$ and $U_j=\bigcup_{r=1}^\ell U_j^{(r)}$, we have
   \[
   K'=\sum_{r=1}^{\ell} t_j \ge |U_j|-\ell \ge K.
   \]
   For all $r=1,\dots,\ell$, we have
   \begin{align}
    \Pr[|U_j\ind{r}\cap J|\le 1]
    &= (1-p)^{t_r}(1+t_rp)
   \label{}
   \end{align}
   This is because the probability is exactly $(1-p)^{t_r+1} + (t_r+1)p(1-p)^{t_r}$ when $t_r\ge 1$, and is exactly 1 when $t_r=0$.

   By definition, $U_j\ind{1},\dots,U_j\ind{\ell}$ are disjoint. So the events that $|U_j\ind{r}\cap J|\ge 2$ are independent.
   Thus, the probability that $j\in A_J$ is
   \begin{align}
    \Pr[j\in A_J]
    =\Pr\left[\bigvee_{r=1}^\ell |U_j\ind{r}\cap J|\geq 2\right]
    = 1 - \prod_{r=1}^{\ell} \Pr[|U_j\ind{r}\cap J|\le 1]
    &= 1 - \prod_{r=1}^{\ell} (1-p)^{t_r}(1+t_rp)\\
    &= 1 - (1-p)^{K'}\prod_{r=1}^{\ell} (1+t_rp).
   \end{align}

   We now bound this below by a constant.
   First consider the case $K'\le \ell$.
   Then $ \ell \ge K$ and $p\ge \frac{1}{2\sqrt{K}}\ge \frac{1}{2\sqrt{K'}}$.
   When $x_1,\dots,x_\ell$ are constrained to be nonnegative integers with a fixed sum, if there exists $x_i \le x_j-2$, we can strictly increase the product $f(x_1,\dots,x_\ell)\coloneqq \prod_{r=1}^{\ell} (1+x_rp)$ by replacing $x_i$ with $x_i+1$ and $x_j$ with $x_j-1$.
   Thus, the maximum value of $f(x_1,\dots,x_\ell)$ occurs when $K'$ of the $x_i$ are 1 and the rest are zero.
   Hence, we have
   \begin{align}
    \Pr[j\in A_J]
    \ge 1 - (1-p)^{K'}(1+p)^{K'}
    = 1-(1-p^2)^{K'}
    \ge 1 - (1-\frac{1}{4K'})^{K'}
    \ge 1 - e^{-1/4},
   \end{align}
   as desired.

   Now suppose $K' > \ell$.
   Note that $p = \min\{\sqrt{\ell}/(2K),\frac{1}{2}\} \ge \sqrt{\ell}/(2K')$.
   As $\log (1+xp)$ is concave for nonnegative real numbers $x$, we have that $f(x_1,\dots,x_\ell)=\prod_{r=1}^{\ell} (1+x_rp)$ subject to $x_1+\cdots+x_\ell=K'$ is maximized when all the $x_i$'s are equal.
   Hence,
   \begin{align}
    \Pr[j\in A_J]
    = 1 - (1-p)^{K'}\prod_{r=1}^{\ell} (1+t_rp)
    &\ge 1 - (1-p)^{K'}\left( 1 + \frac{K'}{\ell}p \right)^{\ell}  \nonumber\\
    &= 1 - \left(\left((1-p)^{K'/\ell}\right)\left( 1 + \frac{K'}{\ell}p \right)\right)^{\ell}
   \end{align}
   as desired.
   If $x\coloneqq K'/\ell\le 1/p$, then, by Lemma~\ref{lem:calc}, we have
   \begin{align}
    \Pr[j\in A_J]
    \ge 1 - \left(1-p^2\frac{(K')^2}{8\ell^2}\right)^\ell
    \ge 1-\left( 1-\frac{1}{32\ell} \right)^{\ell} \ge 1 - e^{-1/32}
   \end{align}
   where the second inequality uses the fact $p\ge \sqrt{\ell}/(2K')$.
   If $x\ge 1/p$, then by Lemma~\ref{lem:calc}, we have $\Pr[j\in A_J] \ge 1-(2/e)^\ell \ge 1-2/e$.
   In all cases,  $\Pr[j\in A_J]$ is bounded below by a constant, as desired.
 \end{proof}

\begin{corollary}\label{cor_existJ}
There exists $J\subseteq [t]$ of cardinality at most $d_1 \sqrt{\ell}t/K$ such that the cardinality of the set $A_J$ as defined in Lemma~\ref{lem_sampling} is at least $d_2 |A|$, where $d_1,d_2>0$ are absolute constants.
\end{corollary}

\begin{proof}
Choose a random set $J\subseteq [t]$ as in  Lemma~\ref{lem_sampling}. Then $\ex[|A_J|]=\Omega(|A|)$ by Lemma~\ref{lem_sampling}.
As $|A_J|\leq |A|$, we have $\Pr[|A_J|\geq d_2|A|]\geq d_3$ for some absolute constants $d_2,d_3>0$.

Observe that by Lemma~\ref{lem_sampling} and the linearity of expectation we have $\ex[|J|]=pt=O(\sqrt{\ell}t/K)$. Moreover, by Markov's inequality we have $\Pr[|J|>d_1 \sqrt{\ell}t/K]\leq d_3/2$ for some sufficiently large constant $d_1>0$. By the union bound, we know the conditions $|J|\leq d_1 \sqrt{\ell}t/K$ and $|A_J|\geq d_2|A|$ are simultaneously satisfied with probability at least $d_3/2>0$, so there exists $J\subseteq [t]$ that satisfies these two conditions.
\end{proof}

Fix $J\subseteq [t]$ as in Corollary~\ref{cor_existJ}, so that $|J|\leq d_1\sqrt{\ell}t/K$ and $|A_J|\geq d_2 |A|$.
As the constant $c$ in \eqref{eq_weightbound} is large enough, we may assume $c\geq d_0 d_1/d_2$, where $d_0$ is as in \eqref{eq_boundA}.
Then we have
\begin{equation}\label{eq_AJ_bound}
|A_J|\geq d_2|A|\stackrel{\eqref{eq_boundA}}{\geq} d_2  \cdot \frac{\mathrm{wt}_\ell(I_{[t]})}{d_0 K (\log(\frac{1}{\epsilon})+\log( \frac{1+\delta}{\delta })+1)}  \stackrel{\eqref{eq_weightbound}}{\geq} (d_1\sqrt{\ell}t/K)k > (|J|-1)k.
\end{equation}

For each $j\in A_J$, choose $c_j\in S_j$ and $T_j\subseteq J$ such that $|T_j|=2$ and $c_{i,j}=c_j$ for $i\in T_j$. This is possible by the definition of $A_J$ in Lemma~\ref{lem_sampling}. For $j\in [n]\setminus A_J$, let $T_j=\emptyset$.  So for $j\in [n]$, we have
\[
|T_j|=\begin{cases} 2 & j\in A_J,\\
0 & j\not\in A_J.
\end{cases}
\]
Let $I_i'=\{j\in [n]: i\in T_j\}\subseteq I_i$ for $i\in J$. We have
\[
\mathrm{wt}(I_J')=\sum_{j=1}^n \max\{|T_j|-1, 0\}=|A_J|\stackrel{\eqref{eq_AJ_bound}}{\geq} (|J|-1)k.
\]
Moreover, the fact $|T_j|\leq 2$ for $j\in [n]$ implies that $I'_{i}\cap I'_{i'}\cap I'_{i''}=\emptyset$ for distinct  $i,i',i''\in J$, by noting that $T_j=\{i\in J: j\in I'_i\}$.

For $i\in [t]$ and $j\in I'_i$, we have $c_{i,j}=c_j$ by the definition of $I'_i$. In particular, $c_{i,j}=c_{i',j}$ for $i,i'\in [t]$ and $j\in I'_i\cap I'_{i'}$.

Finally, we have $|J|\geq 2$ as $|A_J|\geq d_2 |A|>0$. To summarize, we have proved the following weaker version of Lemma~\ref{lem_subset}.
\begin{lemma}\label{lem_weak_subset}
Under the assumption and notations of Lemma~\ref{lem_subset}, there exist $J\subseteq [t]$ and a collection $(I_i')_{i\in J}$ of subsets of $[n]$ such that $|J|\geq 2$, $I_i'\subseteq I_i$ for $i\in J$, and the following conditions are satisfied:
\begin{enumerate}[(1)]
\item $I'_{i}\cap I'_{i'}\cap I'_{i''}=\emptyset$ for distinct  $i,i',i''\in J$.
\item $\mathrm{wt}(I'_J)\geq (|J|-1)k$.
\item $c_{i,j}=c_{i',j}$ for $i,i'\in [t]$ and $j\in I'_i\cap I'_{i'}$.
\end{enumerate}
\end{lemma}

Now we are ready to prove Lemma~\ref{lem_subset},
which has already been partially proven as Lemma~\ref{lem_weak_subset}.
The only part left is to strengthen  $\mathrm{wt}(I'_J)\geq (|J|-1)k$ to $\mathrm{wt}(I'_J)=(|J|-1)k$ and to meet the condition that $\mathrm{wt}(I'_{J'})\leq (|J'|-1)k$ for all nonempty $J'\subseteq J$.

\begin{proof}[Proof of Lemma~\ref{lem_subset}]
Choose the sets $J$ and $(I_i')_{i\in J}$ satisfying Lemma~\ref{lem_weak_subset} such that $|J|\geq 2$ is minimized.
Note that removing one element from $I'_i$ for some $i\in J$ preserves (1) and (3) of Lemma~\ref{lem_weak_subset} and reduces $\mathrm{wt}(I'_J)$ by at most one.
Removing elements from the sets in $(I_i')_{i\in J}$ one by one until $\mathrm{wt}(I'_J)=(|J|-1)k$ holds. Then $J$ and $(I_i')_{i\in J}$ satisfy (1), (3), and (4) of Lemma~\ref{lem_subset}.

The minimality of $|J|$ guarantees that $\mathrm{wt}(I'_{J'})\leq (|J'|-1)k$ for all nonempty $J'\subseteq J$. (When $|J'|=1$, this holds since $\mathrm{wt}(I'_{J'})=0$.) So  $J$ and $(I_i')_{i\in J}$  satisfy (2) of Lemma~\ref{lem_subset} as well.
\end{proof}

\section{Towards Conjecture~\ref{conjecture-0}: A Hypergraph Nash-Williams--Tutte Conjecture}
\label{sec:conjectures}

Recall that Conjecture~\ref{conjecture-0} states that RS codes of rate $R$ are list-decodable from radius $1 - R - \eps$ with list size at most $\lceil\frac{1 - R - \eps}{\eps}\rceil$.
As discussed in the introduction, it was shown in \cite[Theorem 5.8]{shangguan2019combinatorial0} that resolving Conjecture~\ref{conjecture} (about the non-singularity of intersection matrices) would resolve Conjecture~\ref{conjecture-0} (about list-decoding).

Our approach above was to show in Theorem~\ref{theorem-weaker} that intersection matrices are nonsingular under the additional assumption that $I_i \cap I_j \cap I_\ell = \emptyset$ for all $1\le i < j < \ell\le t$, and then use that to conclude our main result about list-recovery.
However, to prove Conjecture~\ref{conjecture-0} in full, we need to remove the additional three-wise intersection assumption.
In this section, we describe an approach that could potentially resolve Conjecture~\ref{conjecture-0} in full.
To do so, we conjecture a hypergraph generalization of the Nash-Williams--Tutte theorem~\cite{nash1961edge,tutte1961problem} that may be of independent interest, and prove that this conjecture implies Conjecture~\ref{conjecture-0}.
Indeed, one might hope that such a generalization would be useful for Conjecture~\ref{conjecture-0}, because the Nash-Williams--Tutte theorem has been instrumental in proving several of the results in this paper, including Theorem~\ref{thm:main-LD}, Theorem~\ref{thm:main}, Theorem~\ref{theorem-weaker}, and Theorem~\ref{perfect-hash-matrix}.
Along the way, we give a second proof of our main list decoding theorem, Theorem~\ref{thm:main-LD}, by combining this approach with known results on hypergraph packings \cite{CS07, CCV09}
In Section~\ref{subsec:intro-hyp}, we describe our conjecture, Conjecture~\ref{conj:weak}, and in Section~\ref{sec:complete}, we prove that it implies the optimal list decoding conjecture, Conjecture~\ref{conjecture-0}.

\subsection{A Hypergraph Nash-Williams--Tutte Conjecture}
\label{subsec:intro-hyp}
In this section we state a conjecture (Conjecture~\ref{conj:weak} below), and prove that it implies our main goal Conjecture~\ref{conjecture-0}.
We are not able to prove Conjecture~\ref{conj:weak}, but give some evidence for it, pointing out that special cases and relaxations are known to be true.

Throughout, we use $t$ as the number of vertices in a (hyper)graph.
This variable corresponds to the same $t$ used in $t$-wise intersection matrices.
A (multi)graph $G$ is called {\it $k$-partition-connected} if every partition $\mathcal{P}$ of the vertex set has at least $k(|\mathcal{P}|-1)$ edges crossing the partition.
By the Nash-Williams--Tutte theorem, this is equivalent to the graph having $k$ edge-disjoint spanning trees.
The parameter $k$ here is the same $k$ used as the dimension of the Reed--Solomon code and the same $k$ used for the Vandermonde matrix degrees in the intersection matrices.

We say a hypergraph $H$ is \emph{$k$-weakly-partition-connected}\footnote{There is also a notion of ``$k$-partition-connected'' for hypergraphs which uses $\min\{\mathcal{P}(e)-1,1\}$ in the sum. In other words, a hypergraph is $k$-partition-connected if any partition $\mathcal{P}$ has at least $k(|\mathcal{P}|-1)$ crossing edges. This notion admits a Nash-Williams--Tutte type theorem: any $k$-partition-connected hypergraph can be decomposed into $k$ 1-partition-connected hypergraphs \cite{frank2003decomposing}} if, for every partition $\mathcal{P}$ of the vertices of $H$, we have
  \begin{align}
    \label{eq:weak-part}
    \sum_{e\in E(H)}^{} (\mathcal{P}(e)-1)\ge k(|\mathcal{P}|-1),
  \end{align}
  where $\mathcal{P}(e)$ is the
number of parts of $\mathcal{P}$ that $e$ intersects.
For example, any $k$-partition-connected graph is $k$-weakly-partition-connected as a hypergraph: \black{for graph edges $e$, the summand $\mathcal{P}(e)-1$ is 1 if edge $e$ crosses the partition and 0 otherwise}. As another example, $k$ copies of a hyperedge covering all $t$ vertices of $H$ is also $k$-weakly partition-connected.

An \emph{edge-labeled} graph is a graph $G$ where each edge is assigned a label from some set $E$.
Let $H$ be a hypergraph.
A \emph{tree-assignment} of $H$ is an edge-labeled graph $G$ obtained by replacing each edge $e$ of $H$ with a tree $F_e$ of $|e|-1$ edges on the vertices of $e$.
\black{Note that, in general, a hypergraph has many possible tree assignments.}
Furthermore, each edge of the graph $F_e$ is labeled with $e$.
The graph $G$ is thus the union of the graphs $F_e$ for $e\in H$.

  A \emph{$k$-tree-decomposition} of a graph on $k(t-1)$ edges is a partition of its edges into $k$ edge-disjoint spanning trees $T_0,\dots,T_{k-1}$.
  We say \emph{tree-decomposition} when $k$ is understood.
  \black{In an edge-labeled graph $F$ with edge-labels from some set $E$, let $v^F\in\mathbb{N}^{E}$ be the vector counting the edge-labels in $T$.
  Specifically, $v^F_e$ is the number of edges of label $e$ in $F$.}
  For a tree-decomposition $(T_0,\cdots,T_{k-1})$ of an edge-labeled graph, define its \emph{signature} $v^{(T_0,\dots,T_{k-1})}$ by
  \begin{align}
    \label{eq:signature}
    v^{(T_0,\dots,T_{k-1})}\coloneqq \sum_{i=0}^{k-1} i\cdot v^{T_i}.
  \end{align}
  An edge-labeled graph $G$ on $t$ vertices is called \emph{$k$-distinguishable} if $G$ has $k(t-1)$ edges and there exists a tree-decomposition $T_0,\dots,T_{k-1}$ of $G$ with a unique signature.
  That is, for any tree-decomposition $T_0',\dots,T_{k-1}'$ with the same signature $v^{(T_0',\dots,T_{k-1}')}= v^{(T_0,\dots,T_{k-1})}$, we have $T_i'=T_i$ for $i=0,\dots,k-1$.

\usetikzlibrary{calc, positioning, fit}
\tikzstyle{vertex} = [fill,shape=circle,node distance=50pt,scale=0.7,font=\tiny,label={[font=\tiny]}]
\tikzstyle{edge} = [opacity=1,fill opacity=0,line cap=round, line width=2pt]
\tikzstyle{hyperedge} = [fill,opacity=0.5,fill opacity=0,line cap=round, line join=round, line width=17pt]
\tikzstyle{he}=[draw, rounded corners, line width=4pt, inner sep=0pt]        

\pgfdeclarelayer{background}
\pgfsetlayers{background,main}

With these definitions, we can now conjecture a hypergraph version of the Nash-Williams--Tutte theorem.

\begin{conjecture}\label{conj:weak}
  Let $t$ and $k$ be positive integers.
  Every $k$-weakly-partition-connected hypergraph $H$ on $t$ vertices has a $k$-distinguishable tree-assignment.
\end{conjecture}

The key result of this section is that our optimal list decoding of Reed--Solomon codes conjecture follows from our hypergraph Nash-Williams--Tutte conjecture.
This connection is proved in Section~\ref{sec:complete} below.
\begin{theorem}
  \label{thm:complete-2}
  Conjecture~\ref{conj:weak} implies Conjecture~\ref{conjecture} and thus Conjecture~\ref{conjecture-0}.
\end{theorem}

One should convince themselves that Conjecture~\ref{conj:weak} (if true) is a generalization of the Nash-Williams--Tutte theorem.
Indeed, when $H$ is a (non-hyper) graph, the conjecture boils down to the Nash-Williams--Tutte theorem.
If $H$ is a graph on $k(t-1)$ edges, then there is only one tree-assignment $G$ of $H$, namely $H$ itself with each edge labeled by itself.
All edges have distinct edge-labels, so for any tree-decomposition $T_0,\dots,T_{k-1}$, the signature $v^{(T_0,\dots,T_{k-1})}$ is unique.
Thus, showing $G$ is distinguishable, is equivalent to showing $G$ has $k$ edge-disjoint spanning trees, which follows from the Nash-Williams--Tutte theorem.
In the correspondence between hypergraph partitions and intersection matrices, the special case when $H$ is a graph corresponds to Theorem~\ref{theorem-weaker}.

To give more intuition when $H$ is not a graph, we give the following example.
\begin{example}\label{ex:1}
  Let $t=4$, and $k=2$.
  Below, $H$ is a $2$-weakly-partition-connected hypergraph.
  We take a tree assignment of $H$ to obtain an edge-labeled graph $G$ on 6 edges, where each edge is labeled by its color.
  The tree-decomposition $T_0\cup T_1$ demonstrates that $G$ is 2-distinguishable:
  We have $v^{T_0} = (\green{2},\orange{1},\magenta{0})$ and $v^{T_1}=(\green{0},\orange{1},\magenta{2})$ so the signature is $v^{(T_0,T_1)} = 0\cdot v^{T_0}+1\cdot v^{T_1} = (\green{0},\orange{1},\magenta{2})$.
  One can check that any other tree decomposition $(T_0',T_1')$ of $G$ has a different signature  $v^{(T_0',T_1')} \neq (\green{0},\orange{1},\magenta{2})$.
  Thus, $G$ is $2$-distinguishable.
  Hence, $H$ is a 2-weakly partition-connected hypergraph with a 2-distinguishable tree-assignment, as predicted by Conjecture~\ref{conj:weak}.
  \begin{center}
    \begin{tikzpicture}[scale=0.4]
      \begin{scope}
        \node[] at (0,5) {$H$};
        \node[vertex,label=above:\(\)] (1) at (0,0) {};
        \node[vertex,label=above:\(\)] (2) at (90:3) {};
        \node[vertex,label=above:\(\)] (3) at (210:3) {};
        \node[vertex,label=above:\(\)] (4) at (330:3) {};

        \begin{pgfonlayer}{background}
        \draw[hyperedge,color=green] (1.center) -- (2.center) -- (3.center)--(1.center);
        \draw[hyperedge,color=orange] (1.center) -- (4.center) -- (3.center)--(1.center);
        \draw[hyperedge,color=magenta] (1.center) -- (2.center) -- (4.center)--(1.center);
        \end{pgfonlayer}
      \end{scope}
      \begin{scope}[xshift=10cm]
        \node[] at (0,5) {$G$};
        \coordinate (1) at (0,0) {};
        \coordinate (2) at (90:3) {};
        \coordinate (3) at (210:3) {};
        \coordinate (4) at (330:3) {};
        \node[vertex,label=above:\(\)] at (1) {};
        \node[vertex,label=above:\(\)] at (2) {};
        \node[vertex,label=above:\(\)] at (3) {};
        \node[vertex,label=above:\(\)] at (4) {};

        \begin{pgfonlayer}{background}
        \draw[edge,color=green] ($ (2) + (180:0.15)$) -- ($ (1) + (150:0.15)$) -- ($ (3) + (90:0.15)$) ;
        \draw[edge,color=orange] ($ (3) + (0:0.15)$) -- ($ (1) + (270:0.15)$) -- ($ (4) + (180:0.15)$) ;
        \draw[edge,color=magenta] ($ (2) + (0:0.15)$) -- ($ (1) + (30:0.15)$) -- ($ (4) + (90:0.15)$) ;
        \end{pgfonlayer}
      \end{scope}
      \begin{scope}[xshift=5cm]
        \node[] at (0,1) {$\longrightarrow$};
      \end{scope}
      \begin{scope}[xshift=15cm]
        \node[] at (0,1) {$=$};
      \end{scope}
      \begin{scope}[xshift=24cm]
        \node[] at (0,1) {$+$};
      \end{scope}
      \begin{scope}[xshift=20cm]
        \node[] at (0,5) {$T_0$};
        \coordinate (1) at (0,0) {};
        \coordinate (2) at (90:3) {};
        \coordinate (3) at (210:3) {};
        \coordinate (4) at (330:3) {};
        \node[vertex,label=above:\(\)] at (1) {};
        \node[vertex,label=above:\(\)] at (2) {};
        \node[vertex,label=above:\(\)] at (3) {};
        \node[vertex,label=above:\(\)] at (4) {};

        \begin{pgfonlayer}{background}
        \draw[edge,color=green] (2)--(1)--(3);
        \draw[edge,color=orange] (1)--(4);
        \end{pgfonlayer}
      \end{scope}
      \begin{scope}[xshift=28cm]
        \node[] at (0,5) {$T_1$};
        \coordinate (1) at (0,0) {};
        \coordinate (2) at (90:3) {};
        \coordinate (3) at (210:3) {};
        \coordinate (4) at (330:3) {};
        \node[vertex,label=above:\(\)] at (1) {};
        \node[vertex,label=above:\(\)] at (2) {};
        \node[vertex,label=above:\(\)] at (3) {};
        \node[vertex,label=above:\(\)] at (4) {};

        \begin{pgfonlayer}{background}
        \draw[edge,color=magenta] (2)--(1)--(4);
        \draw[edge,color=orange] (1)--(3);
        \end{pgfonlayer}
      \end{scope}
    \end{tikzpicture}
  \end{center}
\end{example}

As evidence towards Conjecture~\ref{conj:weak}, we point out that the ``easy part'' of the conjecture follows from the Nash-Williams--Tutte theorem.
One can check that, even in general hypergraphs, every tree-assignment of a $k$-weakly-partition-connected hypergraph gives a $k$-partition connected graph.
Thus, any tree-assignment graph can be partitioned into $k$ spanning trees by the Nash-Williams--Tutte theorem, establishing the ``existence'' part of the graph $G$ being $k$-distinguishable.
Thus, the hard part of Conjecture~\ref{conj:weak} is the ``uniqueness'' part, finding a spanning tree partition with a unique signature.

As further evidence towards Conjecture~\ref{conj:weak}, we point out that a relaxation of Conjecture~\ref{conj:weak} is true by assuming a larger weak-partition-connectivity \cite{CS07, CCV09}.
\begin{theorem}[Follows from \cite{CS07,CCV09}]
  \label{thm:conj-weak}
  There exists an absolute constant $C$ such that the following holds.
  Let $t$ and $k$ be positive integers.
  Every $C(\log t)k$-weakly-partition-connected hypergraph $H$ on $t$ vertices has a tree-assignment with a $k$-distinguishable subgraph.
\end{theorem}

The results in \cite{CS07,CCV09} look slightly different from Theorem~\ref{thm:conj-weak} and we briefly describe the connection here.
Both of the works prove that every $C(\log t)k$-weakly partition-connected hypergraph has $k$ hyperedge-disjoint connected subhypergraphs:\footnote{A hypergraph is connected if, for every two vertices $v$ and $v'$, there is a path $v=v_0,v_1,\dots,v_\ell=v'$ such that, for all $i$, $v_{i-1},v_i$ share a hyperedge.}
the first \cite{CS07} proves an equivalent statement about bipartite Steiner tree packings, and the second \cite{CCV09} proves a more general result about disjoint bases of polymatroids and also improves the constant $C$ in front of the log factor over \cite{CS07}.
Furthermore, it is not difficult to show that, for any hypergraph $H$ with $k$ hyperedge-disjoint connected subhypergraphs $H_0,\dots,H_{k-1}$, any tree assignment of $H$ has a $k$-distinguishable subgraph: briefly, the $k$-distinguishable subgraph $G$ is the union of spanning trees $T_0,\dots,T_{k-1}$ of the tree assignments of $H_0,\dots,H_{k-1}$ that are implied by the tree assignments of $H$, and the spanning trees $T_0,\dots,T_{k-1}$ give the $k$ edge-disjoint spanning trees of $G$ that certify $k$-distinguishability.
Hence, any $C(\log t)k$-weakly-partition-connected hypergraph has $k$ hyperedge-disjoint connected subhypergraphs and thus is $k$-distinguishable, which is Theorem~\ref{thm:conj-weak}.

While it is known that the log factor in the results of \cite{CS07,CCV09} cannot be removed \cite{BT03}\footnote{\cite{BT03} shows there exist $\Omega(\log n)$-weakly-partition-connected hypergraphs without two edge-disjoint connected subhypergraphs.}, one can still hope that Conjecture~\ref{conj:weak} is true.
The results in \cite{CS07, CCV09} are used to obtain $k$-distinguishability by taking a tree packing where edges of the same label are always in the same tree.
Keeping the same edge labels in the same trees is not necessary in general to obtain distinguishability.
Indeed Example~\ref{ex:1} illustrates that a tree-assignment can be $k$-distinguishable even when the original hypergraph does not have a partition into $k$ connected subhypergraphs.
Thus, we hope that the log factor can still be saved by splitting edge labels of the tree-assignment graph across the $k$ spanning trees.
As we pointed out earlier, the existence of a tree-packing in the tree-assignment graph follows from the Nash-Williams--Tutte theorem, so the hard part of the conjecture is not finding the tree-packing, but finding a signature-unique one.

We point out that Theorem~\ref{thm:conj-weak} can be used to obtain a second proof of Theorem~\ref{thm:main-LD} by following the proof of Theorem~\ref{thm:complete-2}, which connects hypergraph partitions to intersection matrices and then applies a polynomial method to obtain list decodable Reed--Solomon codes.
The extra $O(\log t)$ factor in the weak-partition-connectivity of Theorem~\ref{thm:conj-weak} appears as an $O(\log\frac{1}{\varepsilon})$ factor loss in the rate of the Reed--Solomon code.
Since one proof of Theorem~\ref{thm:main-LD} is already given, and the key ideas for this second proof are covered throughout the rest of the paper, we omit the details of this second proof of Theorem~\ref{thm:main-LD}.

In addition to implying optimal list decoding of Reed-Solomon codes (Theorem~\ref{thm:conj-weak}), Conjecture~\ref{conj:weak} may be of independent interest as a candidate hypergraph Nash-Williams--Tutte generalization.
On one hand, a hypergraph Nash-Williams--Tutte generalization is known for \emph{partition-connectivity} \cite{frank2003decomposing}.
A hypergraph is $k$-partition-connected if any partition $\mathcal{P}$ has at least $k(|\mathcal{P}|-1)$ crossing edges.
Frank, Kir\'{a}ly, and Kriesell \cite{frank2003decomposing} showed that every $k$-partition-connected hypergraph has $k$ edge-disjoint 1-partition-connected subhypergraphs.
On the other hand, Nash-Williams--Tutte generalizations for \emph{weak-partition-connectivity} seem to be less studied, and Conjecture~\ref{conj:weak} provides a plausible generalization for weak-partition-connectivity.

\subsection{Proof of Theorem~\ref{thm:complete-2}}\label{sec:complete}

In this section we prove Theorem~\ref{thm:complete-2}.
To do so, we show the first implication, establishing the connection between hypergraph partitions and intersection matrices outlined in Section~\ref{subsec:intro-hyp}.
The second implication was proved in \cite{shangguan2019combinatorial0}.
At a high level, the first implication of Theorem~\ref{thm:complete-2} holds because the uniqueness of the signature of $k$ edge-disjoint spanning trees implies the uniqueness of a monomial in the determinant expansion of an intersection matrix.
Because such a monomial is unique, it does not cancel with any other terms in the determinant expansion, implying that the determinant is nonzero.
We now give the details.

We first derive a sufficient condition for a hypergraph being $k$-weakly-partition-connected.

\begin{lemma}
  \label{lem:rand-2}
  Let $H$ be hypergraph on the vertex set $[t]$ where for all $J\subsetneq [t]$,
  \begin{align}
    \label{eq:weak-part-2}
    \sum_{e\in E(H)}^{} \max(0,|e\cap J|-1) \ &\le \   k(|J|-1) \nonumber\\
    \text{~and~} \sum_{e\in E(H)}^{} (|e|-1) \ &\ge \   k(t-1).
  \end{align}
  Then $H$ is $k$-weakly-partition-connected.
\end{lemma}

\begin{proof}
According to \eqref{eq:weak-part} it suffices to show that for any partition $\mathcal{P}$ of the vertices of $H$, $\sum_{e\in E(H)}(\mathcal{P}(e)-1)\ge k(|\mathcal{P}|-1)$.
To see this, assume that $\mathcal{P}=\{V_1,\ldots,V_s\}$. Then $\sum_{i=1}^s |V_i|=t$, and for each $e\in E(H)$, $|e|=\sum_{i=1}^s |e\cap V_i|$. By the last equality, it is not hard to check that
\begin{equation*}
    |e|=\mathcal{P}(e)+\sum_{i=1}^s \max\{0,|e\cap V_i|-1\}.
\end{equation*}
It follows that
\begin{equation*}
\begin{aligned}
    \sum_{e\in E(H)} (\mathcal{P}(e)-1)=&\sum_{e\in E(H)} \Big(|e|-\sum_{i=1}^s \max\{0,|e\cap V_i|-1\}-1\Big)\\
    &= \sum_{e\in E(H)} (|e|-1)-\sum_{i=1}^s\sum_{e\in E(H)} \max\{0,|e\cap V_i|-1\}\\
    &\ge k(t-1)-\sum_{i=1}^s k(|V_i|-1)=k(s-1),
\end{aligned}
\end{equation*}
where the last inequality follows from \eqref{eq:weak-part-2}.
\end{proof}

We now present the proof of Theorem~\ref{thm:complete-2}.

\begin{proof}[Proof of Theorem~\ref{thm:complete-2}]
Assume Conjecture~\ref{conj:weak} is true.
Let $I_1,\dots,I_t\subseteq[n]$ be subsets satisfying the conditions of Conjecture~\ref{conjecture}.
For all $i\in [n]$, let $e_i=\{j\in [t]: i\in I_j\}$.
Let $H$ be a (multi)hypergraph with vertex set $[t]$ and edge set $E(H)=\{e_i:i\in[n]\}$.

\begin{claim}\label{lem:new-1}
For all subsets $J\subseteq[t]$, we have $\sum_{e\in E(H)}^{} \max(0,|e\cap J|-1)=\wt(I_j:j\in J)$.
\end{claim}

\begin{proof}
  We have
  \begin{equation*}
      \begin{aligned}
          \sum_{i\in [n]} \max\{0,|e_i\cap J|-1\}=\sum_{i\in [n]} |e_i\cap J|-\abs{\bigcup_{j\in J} I_j}=\sum_{j\in J} |I_j|-\abs{\bigcup_{j\in J} I_j}=\wt(I_J),
      \end{aligned}
  \end{equation*}
  where the first equality follows from the fact $\cup_{j\in J} I_j=\{i\in [n]: |e_i\cap J|\ge 1\}$, the second equality follows from easy double-counting, and the last equality follows from \eqref{definition-weight}.
\end{proof}

The following is an easy consequence of Lemmas~\ref{lem:rand-2} and Claim~\ref{lem:new-1}.

\begin{claim}\label{lem:new-2}
  $H$ is $k$-weakly-partitioned-connected.
\end{claim}
\begin{proof}
  By Claim~\ref{lem:new-1} and the setup of Conjecture~\ref{conjecture}, it is clear that for each $J\subsetneq [t]$
\begin{align}
  \sum_{e\in E(H)}^{} \max\{0,|e\cap J|-1\}  \ = \  \wt(I_J) \  &\le \   k(|J|-1) \nonumber\\
  \text{~and~} \sum_{e\in E(H)}^{} \max\{0,|e|-1\} \ = \ \wt(I_{[t]}) \ &\ge \   k(t-1).
\end{align}
It follows by Lemma~\ref{lem:rand-2} that $H$ is $k$-weakly-partition-connected.
\end{proof}

By our assumption that Conjecture~\ref{conj:weak} is true, there exists a tree-assignment $G$ of $H$ that is $k$-distinguishable. Note that by definition $G$ has $k(t-1)$ edges, which are labeled by the hyperedges of $H$. Let $S\subseteq [n]$ be the subset so that $\{e_s:s\in S\}$ forms the set of those labels. Then, an edge $\{j,j'\}$ of $G$ has label $e_s$ for some $s\in S$ if and only if $s\in I_j\cap I_{j'}$.

Recall that
\begin{align}
  M_{k,(I_1,\dots,I_t)} =
  \left(
  \begin{array}{c}
    \mathcal{B}_t\otimes \mathcal{I}_k  \\\hline
  {\rm diag}\Big(V_k(I_j\cap I_{j'}): \{j,j'\}\in \binom{[t]}{2}\Big) \\
  \end{array}
\right),
\label{}
\end{align}
and that the $\binom{t}{2}k$ columns are labeled by the pairs $\{j,j'\}\in\binom{[t]}{2}$, according to the $\binom{t}{2}$ Vandermonde matrices in the bottom diagonal.
Our goal is to show that $M_{k,(I_1,\dots,I_t)}$ is nonsingular. As in the proof of Theorem~\ref{theorem-weaker}, it suffices to show the nonsingularity of

\begin{align}
M'_{k,(I_1,\ldots,I_t)}
\coloneqq
\left(
  \begin{array}{c}
  \mathcal{I}_k\otimes \mathcal{B}_t  \\\hline
  \big(\ma{C}_i:0\le i\le k-1\big) \\
  \end{array}
\right),
\end{align}
where $\ma{C}_i={\rm diag}\Big(V_k^{(i)}(I_j\cap I_{j'}): \{j,j'\}\in \binom{[t]}{2}\Big)$ and $V_k^{(i)}(I_j\cap I_{j'})$ is the $(i+1)$-th column of $V_k(I_j\cap I_{j'})$. Note that $M'_{k,(I_1,\ldots,I_t)}$ is obtained by permuting the columns of $M_{k,(I_1,\ldots,I_t)}$, with the column labels remaining unchanged.

The following fact is easy to verify by definition.

\begin{fact}\label{fact:M'-1}
Each row of $\big(\ma{C}_i:0\le i\le k-1\big)$ has exactly $k$ nonzero entries, which has the form $x_s^0, x_s^1,\dots,x_s^{k-1}$ for some $s\in S$. Moreover, there is $\{j,j'\}\in\binom{[t]}{2}$ so that $s\in I_j\cap I_{j'}$, and those $k$ nonzero entries are all contained in $\{j,j'\}$-labeled columns.
\end{fact}

Let us consider $\binom{t}{2}k\times \binom{t}{2}k$ submatrix $M'$ of $M'_{k,(I_1,\dots,I_t)}$ obtained as follows:

\begin{enumerate}
\item Keep the top $\binom{t-1}{2}k\times \binom{t}{2}k$ submatrix $\mathcal{I}_k\otimes B_t$.
\item For every edge $\{j,j'\}$ in $G$ of label $e_s$, keep the row in $\big(\ma{C}_i:0\le i\le k-1\big)$ with nonzero entries $x_s^0,\dots,x_s^{k-1}$ in $\{j,j'\}$-labeled columns (this is well-defined according to Fact~\ref{fact:M'-1}).
\item Remove all other rows.
\end{enumerate}

\noindent As $G$ has $k(t-1)$ edges, precisely $k(t-1)$ rows are kept in step $2$. Therefore, $M'$ has $\binom{t-1}{2}k + k(t-1) = \binom{t}{2}k$ rows and is thus square.

Below we show that $M'$ has a nonzero determinant, thereby implying that $M'_{k,(I_1,\ldots,I_t)}$ and hence $M_{k,(I_1,\ldots,I_t)}$ are nonsingular, which is the claim of Conjecture~\ref{conjecture}, and thus establishing Theorem~\ref{thm:complete-2}. For that purpose, it is enough to show that there is a monomial that appears as a nonvanishing term in the determinant expansion of $M'$. To find such a monomial, we use the fact that $G$ is $k$-distinguishable.

Recall that for a subgraph $F$ of $G$, we use $v^F\in\mathbb{N}^S$ to denote the vector that counts the  edge labels in $F$, where for $s\in S$, $v_s^F$ is the number of edges with label $e_s$. Note that $v^F$ is a vector of length $|S|$ whose coordinates are indexed by elements in $S$. For spanning trees $T,T_0,\ldots,T_{k-1}$ of $G$, define

\begin{align}
x^T := \prod_{s\in S}^{} x_s^{v^T_s}\text{\qquad and\qquad }x^{(T_0,\dots,T_{k-1})}:=\prod_{i=0}^{k-1} (x^{T_i})^i.
\label{eq:complete-1}
\end{align}

\noindent Observe that, by the definition in \eqref{eq:signature}, we have $v^{(T_0,\ldots,T_{k-1})}_s=\sum_{i=0}^{k-1} i\cdot v_s^{T_i}$ for all $s\in S$. Hence, it follows from \eqref{eq:complete-1} that
\begin{align}\label{equ:new-1}
  x^{(T_0,\dots,T_{k-1})}=  \prod_{s\in S}^{} x_s^{v^{(T_0,\dots,T_{k-1})}_s}.
\end{align}

Since $G$ is $k$-distinguishable, there exists a tree decomposition $T_0,\dots,T_{k-1}$ such that for any other tree-assignment $T_0',\dots,T_{k-1}'$, we have that the signatures $v^{(T_0,\dots,T_{k-1})}\neq v^{(T_0',\dots,T_{k-1}')}$.
Thus, it follows by \eqref{equ:new-1} that the monomials $x^{(T_0,\dots,T_{k-1})}\neq x^{(T_0',\dots,T_{k-1}')}$.

\begin{claim}\label{clm:new-1}
Let $T_0,\ldots,T_{k-1}$ be spanning trees as defined above. Then, $x^{(T_0,\dots,T_{k-1})}$ appears as a nonvanishing term in the determinant expansion of $M'$.
\end{claim}

Proving Claim~\ref{clm:new-1} establishes the nonsingularity of $M'$ and thus, as discussed above, Theorem~\ref{thm:complete-2}. For that purpose, we identify the nonzero entries in the bottom $(t-1)k$ rows of $M'$ by tuples $(s,\{j,j'\},i)$, where $s\in S\cap (I_j\cap I_{j'})$, $\{j,j'\}\in\binom{[t]}{2}$, and $0\le i\le k-1$. Indeed, such a tuple corresponds to the entry $x_s^i$ in the $\{j,j'\}$-labeled column of $\mathcal{C}_i$.
It is worth mentioning that we used two types of labeling here: first, each column of $\mathcal{I}_k\otimes \mathcal{B}_t$ and $(\mathcal{C}_i:~0\le i\le k-1)$ is labeled by some edge $\{j,j'\}\in\binom{[t]}{2}$; second, each edge of $G$ is labeled by some variable $x_s,~s\in S$.

Let $U$ denote the set of all $(t-1)k^2$ nonzero entries in the bottom $(t-1)k$ rows of $M'$. For $Q\subset U$, let $M_Q$ denote the submatrix of $M'$ obtained by removing all of the rows and columns that contain some entry in $Q$. We say $Q$ is a \emph{partial transversal} if it contains exactly one element in each of the bottom $(t-1)k$ rows of $M'$, and no two share a column. By the definition of determinant,
\begin{align}\label{eq:new-2}
  \det(M') = \sum_{Q\text{ partial transversal}}^{} \pm \det(M_Q) \cdot \prod_{(s,\{j,j'\},i)\in Q}^{} x_s^i.
\end{align}

For a partial transversal $Q$ and $0\le i\le k-1$, let $Q_i$ be the subgraph of $G$ that corresponds to the tuples $(s,\{j,j'\},i)\in Q$, namely, $$Q_i=\left\{\{j,j'\}\in \binom{[t]}{2}:~\text{$(s,\{j,j'\},i)\in Q$ for some $s\in I_j\cap I_{j'}$}\right\}.$$ Note that we view each $Q_i$ as a labeled subgraph that preserves the labeling of $G$. Moreover, as $Q$ forms a partial transversal, each $Q_i$ is a simple graph with no multiple edges,  while $G$ could be a multigraph.

We have the following claim.

\begin{claim}\label{clm:new-2}
Let $Q\subseteq U$ be a partial transversal. Then, $\det(M_Q)\neq 0$ if and only if $Q_0,\ldots,Q_{k-1}$ form pairwise edge-disjoint spanning trees of $G$.
\end{claim}

\vspace{5pt}

\noindent {\it Proof of Claim~\ref{clm:new-2}.} By the definition of a partial transversal, $M_Q$ is a $\binom{t-1}{2}k\times \binom{t-1}{2}k$ square matrix with $k$ diagonal blocks, where for each $0\le i\le k-1$, the $(i+1)$-th diagonal block is obtained by removing from $\mathcal{B}_t$ all of the columns that are labeled by the edges of $Q_i$. As $\mathcal{B}_t$ has full row rank, $\det(M_Q)\neq 0$ if and only if each of the $k$ diagonal blocks also has full row rank.
By Claim~\ref{claim-property-of-B_t}, the $(i+1)$-th block has full row rank if and only if the labels of the removed columns form an acyclic graph on the vertices $[t]$, namely, $Q_i$ is acyclic.

Since $k(t-1)$ columns are removed in total, and an acyclic subgraph on $t$ vertices can have at most $t-1$ edges,  $\det(M_Q)\neq 0$ \black{happens if and only if} $Q_i$ is a spanning tree. Moreover, as elements in $Q$ form a partial transversal, we never have $(s,\{j,j,'\},i)$ and $(s,\{j,j'\},i')$ both in $Q$, for $i\neq i'$. It follows that the $Q_i$'s are also pairwise edge-disjoint.
\black{Hence, we have $\det(M_Q)\neq 0$ if and only if the $Q_i$'s form pairwise edge-disjoint spanning trees, as desired.}

\vspace{5pt}

We are now in a position to present the proof of Claim~\ref{clm:new-1}.

\vspace{5pt}

\noindent {\it Proof of Claim~\ref{clm:new-1}.} With the notation above, it is not hard to check by definition  that for a partial transversal $Q\subseteq U$ with $\det(M_Q)\neq 0$, \begin{align}\label{eq:new-3}
    \prod_{(s,\{j,j'\},i)\in Q}^{} x_s^i=\prod_{s\in S}\prod_{i=0}^{k-1} (x_s^{v_s^{Q_i}})^i=\prod_{s\in S} x_s^{v_s^{(Q_0,\ldots,Q_{k-1})}}=x^{(Q_0,\ldots,Q_{k-1})}.
\end{align}

It thus follows from \eqref{eq:new-2}, \eqref{eq:new-3}, and Claim~\ref{clm:new-2} that
\begin{align}
  \det(M')
  = \sum_{Q_0,\dots,Q_{k-1}\text{ edge-disj. spanning trees of $G$}}^{} \pm \det(M_{Q})\cdot x^{(Q_0,\dots,Q_{k-1})}.
\end{align}
It is clear that by the definition of $T_0,\ldots,T_{k-1}$,
in the above summation the monomial $x^{(T_0,\dots,T_{k-1})}$ appears exactly once, and hence appears as a nonvanishing term in the determinant expansion, completing the proof of Claim~\ref{clm:new-1} and thus the proof of Theorem~\ref{thm:complete-2}.
\end{proof}

\section{Future Directions and Open Questions}
In this work, we have shown the existence of near-optimally list-decodable RS codes in the large-radius parameter regime.  To do this, we have established a connection between the intersection matrix approach of \cite{shangguan2019combinatorial0} and tree packings.  Along the way, we also developed applications to the construction of strongly perfect hash matrices, and we have introduced a new hypergraph version of the Nash-Williams--Tutte theorem.
We highlight a few questions that remain open.

\paragraph{Can RS codes \em exactly \em achieve list-decoding capacity, \black{and if so, how large does the field size need to be?}}
In spite of the results and tools developed in this paper, we were not able to prove Conjecture \ref{conjecture-0}.
We hope that the avenue of attack discussed in Section~\ref{sec:conjectures} will be able to finish the job.
We note that the analogous question regarding the limits of  list-recoverability of RS codes also remains open.

\paragraph{Efficient list-decoding of RS codes?} We remark that, using a simple idea from \cite{shangguan2019combinatorial0} one can convert each of the existence results of RS codes reported in this paper into  an explicit code construction, although over a much larger field size. Hence, given such an explicit code construction, is it possible to decode it efficiently up to its guaranteed list-decoding radius? A similar question can be asked for list-recoverability.
We note that \cite{list-dec-discrete-log}, which shows that decoding RS codes much beyond the Johnson bound is likely hard in certain parameter regimes, does not apply to our parameter regime when the field size is large.


\paragraph{Generalizing the Nash-Williams--Tutte theorem to hypergraphs.} In an attempt to resolve Conjecture \ref{conjecture-0}, we present Conjecture~\ref{conj:weak}, a new graph-theoretic conjecture,  which can be viewed as a generalization of the  Nash-Williams--Tutte theorem to hypergraphs. In addition to being interesting on its own, resolving this conjecture would imply the existence of optimally list-decodable RS codes.


\section*{Acknowledgments}

\noindent Zeyu Guo was supported by  NSF-BSF grant CCF-1814629 and 2017732 and the Milgrom family grant for Collaboration between the Technion and the University of Haifa. This work was done while he was at the University of Haifa. He wants to thank Noga Ron-Zewi for helpful discussions.

Ray Li is supported by NSF GRFP grant DGE-1656518 and by Jacob Fox's Packard Fellowship. He thanks Bruce Spang for helpful discussions.

Chong Shangguan is partially supported by the National Key Research and Development Program of China under Grant No. 2021YFA1001000,
the National Natural Science Foundation of China under Grant Nos. 12101364 and 12231014, and the Natural Science Foundation of Shandong Province under Grant
No. ZR2021QA005. Part of the work was done while he was a postdoc at Tel Aviv University and was supported by the Israel Science Foundation (ISF grant number 1030/15).

Itzhak Tamo is partially supported by the European Research Council (ERC grant number 852953), and by the Israel Science Foundation (ISF grant number 1030/15).

Mary Wootters is supported by NSF grant CCF-1844628 and  NSF-BSF grant CCF-1814629, and by a Sloan Research Fellowship.  She thanks Noga Ron-Zewi for helpful discussions.

We thank Karthik Chandrasekaran for helpful discussions about hypergraph packing theorems and for the reference \cite{CCV09}.


\newpage
\appendix

\black{\section{Perfect Hash Matrices and List-recoverable Codes}}

\black{\begin{claim}\label{app:matrices-codes}
    Let $M$ be an $n\times m$ $q$-ary matrix defined on the alphabet set $[q]$. Let $\mathcal{C}$ be a $q$-ary code of length $n$ formed by the columns of $M$. Then, $M$ is $t$-perfect hashing if and only if $\mathcal{C}$ is $(0,t-1,t-1)$-list-recoverable.
\end{claim}
\begin{proof}
    To prove the ``only if'' part, assume for contradiction that $\mathcal{C}$ is not $(0,t-1,t-1)$-list-recoverable. Then, by definition, there exist subsets $S_1,\ldots,S_n\subseteq [q]$, each of size $t-1$, such that
    \begin{align*}
        |\{c\in\mathcal{C}:c_i\in S_i,~1\le i\le n\}|\ge t.
    \end{align*}
    Let $c^1,\ldots,c^t$ be $t$ distinct codewords of $\mathcal{C}$ (or equivalently, columns of $M$) satisfying the condition above. As for each $1\le i\le n$, $\{c^1_i,\ldots,c^t_i\}$ consists of at most $t-1$ distinct symbols of $[q]$, the $i$-th row of $M$ can not separate $c^1,\ldots,c^t$. Hence no row of $M$ can separate $c^1,\ldots,c^t$, a contradiction.

    To prove the ``if'' part, assume for contradiction that $M$ is not $t$-perfect hashing. Then, there exist $t$ distinct columns of $M$ (or equivalently, codewords of $\mathcal{C}$), say $c^1,\ldots,c^t$, that cannot be separated by any row of $M$. In other words, for each $1\le i\le n$, $\{c^1_i,\ldots,c^t_i\}$ consists of at most $t-1$ distinct symbols of $[q]$. It is routine to check by definition that $c^1,\ldots,c^t$ violate the list-recoverability of $\mathcal{C}$, and we have arrived at the desired contradiction.
\end{proof}}

\section{Examples of Intersection Matrices}
\label{app:1}

\begin{example}[$4$-wise intersection matrices]\label{example-1}
Given four subsets $I_1,I_2,I_3,I_4\subseteq [n]$, the 4-wise intersection matrix $M_{k,(I_1,I_2,I_3,I_4)}$ is the  $(3k+\sum_{
1\le i<j\le 4}|I_i\cap I_j|)\times 6k$ variable matrix

\begin{equation*}\label{4-int-matrix-rep-1}
{\small\left(
  \begin{array}{cccccc}
    \mathcal{I}_k &      &     & -\mathcal{I}_k   & \mathcal{I}_k &  \\
        &  \mathcal{I}_k &     & -\mathcal{I}_k   &     & \mathcal{I}_k \\
        &      & \mathcal{I}_k &       & -\mathcal{I}_k & \mathcal{I}_k \\\hline
    V_k(I_1\cap I_2) &  &  &  &  &  \\
     & V_k(I_1\cap I_3) &  &  &  &  \\
     &  & V_k(I_2\cap I_3) &  &  &  \\
     &  &  & V_k(I_1\cap I_4) &  &  \\
     &  &  &  & V_k(I_2\cap I_4) &  \\
     &  &  &  &  & V_k(I_3\cap I_4) \\
  \end{array}
\right).}
\end{equation*}
\end{example}

\begin{example}\label{example-2} \black{Recall that the matrix $\mathcal{B}_t$ is defined below the proof of Lemma \ref{cycle-space-basis}.} For $k=2$, instead of considering the following $4$-wise intersection matrix
  \begin{equation*}
  \left(
  \begin{array}{c}
    \mathcal{B}_4\otimes \mathcal{I}_2  \\\hline

  {\rm diag}(V_k(I_i\cap I_j): \{i,j\}\in \binom{[4]}{2} \\
  \end{array}
\right)=
  \left(\begin{array}{cc|cc|cc|cc|cc|cc}
  1 &  &  &  &  &  & -1 &  & 1 &  &  &  \\
   & 1 &  &  &  &  &  & -1 &  & 1 &  &  \\\hline
   &  & 1 &  &  &  & -1 &  &  &  & 1 &  \\
   &  &  & 1 &  &  &  & -1 &  &  &  & 1 \\\hline
   &  &  &  & 1 &  &  &  & -1 &  & 1 &  \\
   &  &  &  &  & 1 &  &  &  & -1 &  & 1 \\\hline
  1 & x_1 &  &  &  &  &  &  &  &  &  &  \\
   &  & 1 & x_2 &  &  &  &  &  &  &  &  \\
   &  &  &  & 1 & x_3 &  &  &  &  &  &  \\
   &  &  &  &  &  & 1 & x_4 &  &  &  &  \\
   &  &  &  &  &  &  &  & 1 & x_5 &  &  \\
   &  &  &  &  &  &  &  &  &  & 1 & x_6
\end{array}\right),
\end{equation*}
\noindent we turn to prove the nonsingularity of
\begin{equation*}\label{hey}
\left(
  \begin{array}{c}
    \mathcal{I}_2\otimes \mathcal{B}_4  \\\hline

  (\mathcal{C}_i:0\le i\le 1) \\
  \end{array}
\right)=
\left(\begin{array}{cccccc|cccccc}
  1 &  &  & -1 & 1 &  &  &  &  &  &  &  \\
   & 1 &  & -1 &  & 1 &  &  &  &  &  &  \\
   &  & 1 &  & -1 & 1 &  &  &  &  &  &  \\\hline
   &  &  &  &  &  & 1 &  &  & -1 & 1 &  \\
   &  &  &  &  &  &  & 1 &  & -1 &  & 1 \\
   &  &  &  &  &  &  &  & 1 &  & -1 & 1 \\\hline
  1 &  &  &  &  &  & x_1 &  &  &  &  &  \\
   & 1 &  &  &  &  &  & x_2 &  &  &  &  \\
   &  & 1 &  &  &  &  &  & x_3 &  &  &  \\
   &  &  & 1 &  &  &  &  &  & x_4 &  &  \\
   &  &  &  & 1 &  &  &  &  &  & x_5 &  \\
   &  &  &  &  & 1 &  &  &  &  &  & x_6
\end{array}\right).
\end{equation*}
\end{example}

\newpage

{\small\bibliographystyle{alpha}
\bibliography{nash-williams-tutte}}

\end{document}